\documentclass[11pt]{article}
\usepackage{bsm}

\begin{document}

\title{Bounded Simultaneous Messages}
\author[1]{Andrej Bogdanov}
\author[2]{Krishnamoorthy Dinesh}
\author[3]{Yuval Filmus}
\author[3]{Yuval Ishai}
\author[3]{Avi Kaplan}
\author[4]{Sruthi Sekar}
\affil[1]{School of EECS, University of Ottawa}
\affil[2]{Dept. of Computer Science and Engineering, Indian Institute of Technology, Palakkad}
\affil[3]{The Henry and Marylin Taub Faculty of Computer Science, Technion}
\affil[4]{University of California, Berkeley}

\maketitle

\begin{abstract}
We consider the following question of {\em bounded simultaneous messages} (BSM) protocols: Can computationally unbounded Alice and Bob evaluate a function $f(x,y)$ of their inputs by sending polynomial-size messages to a computationally bounded Carol? The special case where $f$ is the mod-2 inner-product function and Carol is bounded to AC$^0$ has been studied in previous works. The general question can be broadly motivated by applications in which distributed computation is more costly than local computation, including secure two-party computation.

In this work, we initiate a more systematic study of the BSM model, with different functions $f$ and computational bounds on Carol. In particular, we give evidence against the existence of BSM protocols with polynomial-size Carol for naturally distributed variants of NP-complete languages. 
\end{abstract}

\section{Introduction} \label{sec:introduction}

The simultaneous messages model 
is a model for evaluating a function $f(x, y)$ of two inputs: Alice, who holds $x$ and Bob, who holds $y$, simultaneously send messages to a referee Carol, who outputs the value $f(x, y)$.  This model was introduced by Yao~\cite{Yao79} and further studied by Babai \textit{et al.}~\cite{BGKL03} in the context of communication complexity, where the interest is to determine the minimum length of Alice's and Bob's messages required to specify $f$.\footnote{Babai \textit{et al.} and others~\cite{BGKL03,NW93,PRS97} more generally study a multi-party variant of this model.}

We are primarily interested in Carol's \emph{computational} complexity, assuming Alice and Bob are unbounded.  Specifically, our objective is to understand which functions require large \emph{bounded} simultaneous messages (BSM) protocols.  As the BSM complexity of any function can only be smaller than its standalone computational complexity, explicit constructions of provably hard functions for this model are subject to the usual barriers~\cite{RR94,AW08}.  Nevertheless, owing to its connections to other complexity notions~\cite{PRS88, Juk06, Lok09},\footnote{In this literature $f$ is the adjacency function of a bipartite graph and BSM complexity is called \emph{bipartite graph complexity}.} explicit almost-cubic lower bounds have been obtained for depth-3 AND-OR formulas~\cite{Lok03}, and almost-quadratic lower bounds for arbitrary AND-OR formulas~\cite{Tal17}, both witnessed by the ``inner product modulo two'' (IP) function.  It remains a challenge to prove BSM lower bounds for other models in which circuit lower bounds are known such as depth-3 AND-OR circuits~\cite{PSZ00} and parities of DNF~\cite{CS16}.

Motivated by cryptographic applications, Rothblum~\cite{Rot12} conjectures that IP has exponential-size bounded-depth unbounded-fan-in BSM complexity. In particular, IP does not admit a BSM protocol with Carol in AC$^0$. Filmus \textit{et al.}~\cite{FIKK20} confirmed the conjecture assuming Alice's message is not much longer than her input.

\paragraph*{BSM and secure two-party computation}
A primary motivation for revisiting the BSM model stems from its relevance to cryptographically secure two-party computation.  In this model, the BSM model is the ``ideal model'' in which evaluation is carried out by a trusted intermediary Carol.  A crowning achievement of modern cryptography is the development of secure two-party protocols that dispense off the need for a trusted Carol~\cite{Yao86,GMW87,BGW88,CCD88}. However, the complexity of the ideal model is of interest for the following reasons.

First, despite tremendous progress in the design of two-party protocols for secure function evaluation, the overhead remains substantial.  Moreover, certain features of trusted evaluation, such as fairness, are inherently lost by any two-party implementation.  In practice, a trusted Carol is often preferred over a two-party protocol, and Carol's computational cost needs to be compensated. BSM complexity captures Carol's minimal computational cost incurred under the best-possible preprocessing of Alice's and Bob's inputs.

Second, in the design of two-party protocols, it can be desirable to preprocess the inputs so as to minimize the complexity of the interactive phase.  The complexity measure to be optimized depends on the type of protocol.  For instance, in protocols based on oblivious transfer~\cite{GMW87,Kil88,IPS08}, the communication complexity and round complexity of the protocol are proportional to the circuit size and circuit depth of $f$, respectively.  In protocols that rely on fully homomorphic encryption~\cite{Gen09,BV14} and avoid an expensive ``bootstrapping'' technique, the computational cost is governed by the degree of $f$ as a polynomial over some underlying finite field. The degree also determines the fraction of corrupted servers that can be tolerated in non-interactive secure computation protocols that employ a set of untrusted servers. In this context, BSM complexity captures the best-possible savings that can be obtained by endowing Alice and Bob with unbounded computational power in the preprocessing phase.  Conversely, a lower bound in the BSM model indicates limits of preprocessing in a two-party protocol design.

With these motivations in mind, we study the BSM complexity of natural classes of functions that may arise in two-party computation, including Boolean circuits, arithmetic circuits, low-degree polynomials, and time-unbounded halting Turing machines.  As some of these classes do not admit explicit lower bounds, any evidence of their hardness must be conditional.  One of our contributions is the development of hardness reductions and completeness results for BSM complexity.

\paragraph*{Generic constructions and non-explicit bounds}
Our primary focus is on Carol's computational complexity, with Alice's and Bob's message size (i.e., communication complexity) as a secondary parameter.  Clearly every function $f\colon \{0, 1\}^n \times \{0, 1\}^n \to \{0, 1\}$ has  a regular circuit size of at most $2^{2n}$ and therefore BSM circuit size at most $2^{2n}$ with message size $n$.   We are interested in the regime where Alice's and Bob's messages have a length of at least $n$.

It is also known that every function has a BSM  circuit (parity of ANDs) of size $O(2^n)$ with message size $2^n$.  This is close to optimal as most functions require BSM circuits of size $\Omega(2^n)$.  In contrast, most functions' circuit size is $\Omega(2^{2n}/n)$.
Although BSM complexity can be much smaller than circuit complexity, every BSM of circuit size $s$ and total message length $\ell$ can be simulated by a circuit of size $\ell \cdot 2^n + s$.

\subsection{Our Results}\label{sec:our results introduction}

\paragraph*{Conditional impossibility of efficient BSM for NP}
As our first result, we study a distributed variant of SAT, which we refer to as \emph{split SAT}:
\[ \SSAT := \set{(\alpha, \beta) : \alpha, \beta \text{ are CNFs and } \alpha \wedge \beta \text{ is satisfiable}}. \]
 In Section~\ref{ssec:split-SAT-and-NP-languages} we show that if $\SSAT$ has BSM circuit complexity $\poly(\abs{\alpha} + \abs{\beta})$ then every language in $\NP$ can be decided by circuit families of size $O(\poly(n)\cdot 2^{n/2})$.  In contrapositive form, no efficient $\BSM$ for $\SSAT$ exists unless $\NP$ has circuits of size $\poly(n)\cdot 2^{n/2}$.  

We conjecture that the following two $\NP$-languages meet this stringent lower bound:
\begin{enumerate}
	\item Turing Machine Acceptance:  The language $ATM_c$ consists of all pairs $(M, x)$ where $M$ is a nondeterministic Turing Machine that accepts $x$ in $\abs{x}^c$ steps.  Here $c$ is a fixed constant; the conjecture is plausible for any $c > 1$.  This language is complete in the sense that if $\mathrm{NTIME}(n^c)$ requires circuit size $s(n)$ then $ATM_c$ requires circuit size $s(n - O(1))$.
	\item Succinct Subset Sum: Here an $n$-bit input $x$ implicitly defines an instance of the {\em subset sum} problem of size $N\gg n$.\footnote{This is a modified variant of the succinct subset sum candidate from the original version of this paper. In the original candidate, an input $x$ of length $n$ defined a subset sum instance consisting of $n$ elements in $\F_{3^m}$. It was speculated that the $O(2^{n/2})$-time subset sum algorithms in the RAM-model~\cite{HS74, SS81, NW21} do not imply circuits of comparable size. However, as pointed out by Ryan Williams, these algorithms can in fact be implemented by circuits of size $\poly(n)\cdot 2^{n/2}$. Our current candidate uses a bigger implicit subset sum instance to avoid these attacks.  } In contrast to the previous ATM candidate, the current one can be conjectured to be {\em hard on average}. More concretely, let $G$ be a ``pseudorandom'' mapping from an $n$-bit input $x$ to a subset sum instance $G(x)=(x_1,...,x_N)$, where $N=10n$, and each $x_i$ is in $\Z_{2^N}$. The yes-instances are those $x$ for which there exist coefficients $a_0, a_1, \cdots, a_N \in \{0, 1\} \subseteq \Z_{2^N}$ such that $a_0 + a_1x_1 + \ldots + a_nx_N = 0$ and not all $a_i$ are 0.
    The function $G$ can be heuristically instantiated in many ways. To give a concrete proposal: let the output of $G(x)$ consist of the sequence of $N^2$ quadratic characters of $x+1,x+2,\ldots,x+N^2$ modulo the first $(n+1)$-bit prime.

\end{enumerate}

As a consequence, in Corollary~\ref{cor:inversion} we show BSM hardness for $\NP$ follows from the existence of very strong worst-case one-way functions.  The required inversion complexity is as large as $\poly(n) \cdot 2^{n/2}$ under a plausible hardness assumption and $\poly(n) \cdot 2^{2n/3}$ unconditionally (for length-preserving functions).  The best-known algorithms for generic function inversion have conjectured worst-case complexity $\poly(n) \cdot 2^{2n/3}$ in the RAM model with preprocessing~\cite{Hel80}. The best unconditional upper bound is $\poly(n) \cdot 2^{3n/4}$~\cite{FN00}.  
This bound was recently adapted to the circuit model in~\cite{MazorP23a,HIW23}.

\paragraph*{BSM and Instance Hiding} 
We consider the notion of instance hiding introduced by Beaver and Feigenbaum~\cite{BF90}. An instance hiding (IH) scheme consists of a primary actor, called Henry, with query access to independent oracles.  To compute a function on a given input, Henry may query the oracles, but in such a way that each oracle learns nothing about the input, meaning that the query distribution each oracle sees depends only on the input length. Beaver and Feigenbaum allow interactive oracle queries to multiple oracles;  we specialize to non-interactive queries and two oracles. 

Fortnow and Szegedy~\cite{FS91} asked whether languages in $\NP$ can have an instance hiding scheme where Henry is a polynomial-sized circuit. They showed that $\SAT$ cannot have $\poly(n)$-sized IH scheme where the two oracles return $1$-bit answers under a standard complexity-theoretic assumption. However, a similar statement with oracles returning $2$ or more bits is unknown. The hardness assumption about IH which we make here is that languages in $\NP$ do not have a $\poly(n)$-sized instance hiding scheme.

In Proposition~\ref{prop:SAT-split-hide-reduces-to-SSAT} we show that if BSM for SSAT has circuit complexity $\poly(\abs{\alpha} + \abs{\beta})$ then every language in $\NP$ has a polynomial-size instance hiding scheme.  We interpret this as additional evidence against the existence of efficient BSM for $\NP$.
The same holds for the split variants of 3-coloring (Proposition~\ref{prop:3COL-split-hide-reduces-to-SCOL}) and partition (Proposition~\ref{prop:PARTITION-split-hide-reduces-to-SPRT}):
\begin{gather*}
\SCOL := \set{(A, B) : A, B \text{ are graphs on the same vertex set, } A \cup B \text{ is 3-colorable}}, \\
\SPRT := \set{(A, B) : \begin{matrix} A, B \text{ are sets of non-negative integers, } \\ A \cup B \text{ can be partitioned into two sets of equal sum} \end{matrix}}.
\end{gather*}

We achieve this via a notion of a reduction that we call \textit{split-hide} reduction (Definition~\ref{def:split-hide-reduction}). For a language $A$, we define a language $B$ which is a ``split variant'' of $A$, such that if $A$ split-hide reduces to $B$ and  $B$ has a BSM protocol with polynomial sized Carol, then $A$ has IH schemes where Henry is polynomial sized (Lemma~\ref{lem:split-hide-reduction-BSM-to-IH}). In addition, if $A$ is $\NP$-complete, then every language in $\NP$ has polynomial sized IH schemes (Corollary~\ref{cor:split-hide-reduction-IH-NP})

Having established this relation between instance hiding and BSM (in Section~\ref{ssec:instance-hiding}), we also obtain a connection between instance hiding and Private Information Retrieval~\cite{CGKS95}  and use this connection (in the form of reduction between the models) to obtain universal lower bounds on IH schemes (Appendix~\ref{ssec:Universal-lower-bound-for-the-instance-hiding-model}). This, in turn, gives a two-way reduction chain between BSM and locally decodable codes (Appendix~\ref{ssec:from-BSM to-LDC-and-vice-versa}). Both of these connections could be of independent interest. 

\paragraph*{BSM and algebraic polynomials}
We now move on to the setting where Carol is an algebraic polynomial over some ring. As explained before in the introduction, this setting is motivated by the design of secure two-party protocols that rely on fully homomorphic encryption~\cite{Gen09,BV14} and avoids expensive bootstrapping (where the computational cost is dependent on the degree of Carol). Every function $f \co \zo^n \times \zo^n \to \zo$ can be represented by a polynomial of degree at most $2n$, and this is also tight. We ask whether this bound can be reduced if we allow preprocessing. To this end, we define $d_f(n, m)$ to be the minimal degree of a polynomial over $\F_2$ such that there exists a BSM protocol for $f$ with preprocessing length of $m$ bits and Carol computing the polynomial, and we define \[ d(n, m) := \max\set{d_f(n, m) \co f \co \zo^n \times \zo^n \to \zo}. \]
We prove a tight bound of $\Theta(n/\log n)$ for the case $m = \poly(n)$.  The upper bound was obtained by Beaver et al.~\cite{BFKR97}.  We show a lower bound for the equality function (Proposition~\ref{prop:degree-reduction-equality-lower-bound}). 

Going ahead, we ask if there are any savings (in terms of degree) in computing the equality function if we allow Carol to be a polynomial over $\Z_k$ where $k$ is composite, followed by a Boolean decision predicate $P:\Z_k\to\{0,1\}$. We show that this is indeed the case for $k=6$ (Proposition~\ref{prop:BSM-via-matching-vector}) using known constructions of matching vectors~\cite{Gro00}. In particular, we give a constant degree Carol, but with superpolynomial ($\exp(\wt{O}(\sqrt{n}))$) message length over $\Z_6$ which should be contrasted with the BSM protocol over $\F_2$ (Proposition~\ref{prop:universal-degree-reduction-upper-bound}) of polynomial message length but with near linear degree.

\paragraph*{BSM, Matrix multiplication and RE languages}
We now turn to the setting where Carol is an arithmetic circuit as well as a Turing machine. We start with the setting of an arithmetic circuit. The problem of interest is the matrix multiplication over reals. A remarkable result of Strassen~\cite{S69} showed a non-trivial algorithm for matrix multiplication of two $n\times n$ matrices over reals in $O(n^{\log_2 7})$ time. Since then, there has been a flurry of improvements (notably~\cite{CW81,DS13,AW21}) where the best value of the exponent in the runtime is denoted by $\omega$. These results also imply an arithmetic circuit performing $O(n^{\omega})$ additions and multiplications over reals computing the matrix product. A natural question is: can there be a BSM protocol where Carol is an arithmetic circuit of size $o(n^{\omega})$ such that the preprocessing can have a non-trivial saving in size of Carol?  We answer this question in negative in Theorem~\ref{thm:matrix-multiplication-with-preprocessing}.

Finally, we turn to the setting where Carol is a Turing machine. In Proposition~\ref{prop:computability-with-preprocessing-RE}, we show that any recursively enumerable $f\in{\sf RE}$ (viewed as a language) has a BSM protocol where Carol computes a decidable language (in ${\sf R}$).  

\paragraph*{Summary of results} The results are summarized in Table~\ref{tbl:result-summary} according to Carol's computational power.
\begin{itemize}
	\item $\BSM(m, s)$ is the class of all functions $f \co \zo^n \times \zo^n \to \zo$ that have a two-party SM protocol of size $s$ (namely, Carol is implemented by a circuit of size $s$) and preprocessing of length $m$ (namely, Alice and Bob send messages of length $m$; we place a ``$\cdot$'' symbol in cases where the result is independent of the corresponding parameter).	
	\item $\IH(a, s)$ is the class of all functions $f \co \zo^n \to \zo$ that have a two-oracle Instance Hiding scheme in which Henry's circuit size is $s$ and each oracle answer is of length $a$.
\end{itemize}
We use $\RE$ and $\mathsf{R}$ to denote the recursively enumerable and recursive languages respectively.

\renewcommand{\arraystretch}{1.25}
\begin{table}[ht]
	\caption{Summary of Results} \label{tbl:result-summary}
	\centering
	\begin{tabular}{l l l}
		\toprule
		{\textbf{Bound}} & {\textbf{Result}} & {\textbf{Reference}} \\
		\toprule
		{Size} & {$\SSAT \in \BSM(\cdot, \poly(n))$ $\implies$ $\NP \seq \SIZE(\poly(n) \cdot 2^{n/2})$} & Thm.~\ref{thm:BSM-for-split-SAT-implies-upper-circuit-bound-for-NP} \\
		& {\phantom{$\SSAT \in \BSM(\cdot, \poly(n))$} $\implies$ No $\wt{O}(2^{2n/3})$-secure length preserving OWF} & Cor.~\ref{cor:inversion} \\
		& {$\SSAT/\SCOL/\SPRT \in \BSM(\cdot, \poly(n))$ $\implies$ $\NP \seq \IH(\cdot, \poly(n))$} & Thm.~\ref{thm:SSAT-SCOL-SPRT-NP-IH} \\
		& {$g(x \oplus y) \in \BSM(\cdot, s)$ $\implies$ $g \in \SIZE((s/n) \cdot 2^{n/2})$} & Prop.~\ref{prop:BSM-to-circuit-upper-bound} \\
		& {$g(x \vee y), g(x \wedge y) \in \BSM(2^{0.729n}, \cdot)$ for all functions $g \co \zo^n \to \zo$} & Thm.~\ref{thm:BSM-upper-bound-for-all-bitwise-combined-functions} \\
		\midrule
		{Arithmetic} & {Matrix multiplication with preprocessing requires $\Omega(n^\omega)$ operations} & Thm.~\ref{thm:matrix-multiplication-with-preprocessing} \\
	      \midrule
		{Degree} & {$d(n,m) = O(n/\log(m/n))$ for $n < m$}  & Prop.~\ref{prop:universal-degree-reduction-upper-bound} \\
		& {$d_\EQ(n, m) = \Omega(n/\log m)$ for $n \leq m$} & Prop.~\ref{prop:degree-reduction-equality-lower-bound} \\
		& {$d(n,n^{C+1}) = \Theta(n/\log n)$ ($C>0$ constant)} & Cor.~\ref{cor:poly-degree-reduction-tight-bound} \\
		& {$d_\EQ(n, m) = 2$ for $m = 2^{\wt{\Theta}(\sqrt{ n})}$} (over $\Z_6$) & Prop.~\ref{prop:BSM-via-matching-vector} \\
		\midrule
		{Recursive} & $\RE \seq \BSM(2n+1, \mathsf{R})$ & Prop.~\ref{prop:computability-with-preprocessing-RE} \\
		\bottomrule
	\end{tabular}
\end{table}

\subsection{BSM protocols implicit in literature}
In Appendix~\ref{ssec:bitwise-combined-functions} we discuss BSM protocols for a subclass of Boolean functions that can be obtained readily from known results in the literature. In particular, we consider the bitwise-combined functions which are functions of the form $f(x,y) = g(x * y)$ where $g$ is a Boolean function on $n$ bits and $* \co \zo^n \times \zo^n \to \zo^n$ is a \textit{combiner function} that takes in two $n$ bit strings and output an $n$ bit string. Babai \textit{et.al}~\cite{BGKL03} considers SM protocols, where $*$ is the bitwise XOR combiner function, with a preprocessing length of $O(2^{0.92n})$. Ambainis and Lokam~\cite{AL00} improved this to obtain SM protocols for bitwise XOR combined functions where the preprocessing length is bounded by $O(2^{0.729n})$. A standard counting argument shows that most XOR combined functions require a preprocessing length of $\Omega(2^{0.5n})$.  (Proposition~\ref{prop:lower-bound-for-most-bitwise-XOR combined-functions}). The XOR combined functions are interesting since an efficient BSM protocol computing $f$ would give an efficient IH scheme for $g$ (Proposition~\ref{prop:bitwise-XOR-BSM-to-IH}). Using this connection (in Section~\ref{sssec:combiners-and-structural-complexity}), we could rule out bitwise-XOR combined version of $\NP$-hard languages from having $\poly(n)$ sized BSM protocols with message complexity of $n+1$ unless the polynomial hierarchy collapse and bitwise-XOR combined version $\PSPACE$-hard languages having $\poly(n)$ sized BSM protocols with a message complexity of $n+O(\log n)$ unless $\PSPACE \subseteq \PP^{\NP}/\poly$. These conclusions rely on structural complexity results about instance hiding due to Fortnow and Szegedy~\cite{FS91} and Tripathi~\cite{Tri10}.

We continue the study where the bitwise combiner function is a bitwise OR as well as bitwise AND function and obtain SM protocols with a preprocessing length of $O(2^{0.729n})$ matching the result of Ambainis and Lokam (Theorem~\ref{thm:BSM-upper-bound-for-all-bitwise-combined-functions}) for XOR combined functions~\cite{AL00}. 

\subsection{Paper organization}
The overview for the rest of the paper is as follows. We study the BSM model under computational restrictions on Carol. In Section~\ref{sec:BSM-for-NP-problems}, we restrict the Carol to be a polynomial-sized circuit, and show that the Split-SAT problem ($\SSAT$) is unlikely to have polynomial sized BSM protocols and we extend the same to two other distributed variants of $\NP$-complete problems. In Section~\ref{sec:arithmetic-circuits-of-bounded-size}, we restrict Carol to be an arithmetic circuit and in Section~\ref{sec:polynomials-of-bounded-degree}, we restrict Carol to be an algebraic polynomial over $\F_2$. In Section~\ref{sec:computability-with-preprocessing}, we discuss the setting where Carol is a Turing machine and show that it is possible to decide any $\RE$ language (including the Halting problem) with preprocessing.

In Appendix~\ref{sec:related-models}, we provide connections of BSM to related models. In particular, we exploit this connection to argue that most Boolean functions do not have efficient instance hiding schemes.
In Appendix~\ref{sec:implicit-BSM-bounds}, we discuss a few BSM protocols that are implicit in the existing literature.

\paragraph*{Acknowledgments}
We thank Rahul Ilango, Ivan Mihajlin, Hanlin Ren, and Ryan Williams for helpful discussions.
A. Bogdanov was supported by an NSERC Discovery Grant.
Y. Filmus was supported by the European Union's Horizon 2020 research and innovation programme under grant agreement No 802020-ERC-HARMONIC.
Y. Ishai was supported by ERC Project NTSC (742754), ISF grant 2774/20, and BSF grant 2018393.
A. Kaplan was supported by the European Union's Horizon 2020 research and innovation programme under grant agreement No 802020-ERC-HARMONIC, and ERC Project NTSC (742754).
S. Sekhar was supported by DARPA under Agreement No. HR00112020026, AFOSR Award FA9550-19-1-0200, NSF CNS Award 1936826, and research grants by the Sloan Foundation, and Visa Inc.

\section{Bounded Simultaneous Messages for NP problems} 
\label{sec:BSM-for-NP-problems}

In this section we show that the existence of BSM for split variants of certain NP complete problems yields unlikely consequences, including better than brute-force circuits for all NP languages (Section~\ref{ssec:split-SAT-and-NP-languages}), one-way functions (Section~\ref{ssec:owf}), and instance hiding schemes for all of NP (Section~\ref{ssec:instance-hiding}).

\subsection{Split SAT and NP languages} \label{ssec:split-SAT-and-NP-languages}

    \begin{theorem} \label{thm:BSM-for-split-SAT-implies-upper-circuit-bound-for-NP}
        If $\SSAT$ has a $2$-party BSM protocol of size $s(n)$ and message length $\ell(n)$ then every $L$ in $\NP$ has circuit size $O(\ell(p(n)) \cdot 2^{n/2}) + s(p(n))$ for some polynomial $p$.
    \end{theorem}

    \begin{proof}
    We will assume without loss of generality that $L$ rejects all inputs of odd length.  This can be enforced using the encoding
    \[ L = \{0x\colon \text{$x$ is a yes instance of odd length}\} \cup \{11x\colon \text{$x$ is a yes instance of even length}\}. \]
    Consider the following algorithm for $L$:

\medskip
    \noindent{\bf Preprocessing}  Given an even input length $n$:
	\begin{enumerate}
	\item Construct the Cook-Levin CNF $\phi$ which accepts $(x, y, z)$ for some $z$  if and only if 
$xy \in L$ with $\abs{x} = \abs{y} = n/2$.
	\item For every $x$ with $\abs{x} = n/2$, construct the CNF  
	\[ \alpha_x(u, v, z) = \text{``}(x = u) \wedge \phi(x, v, z)\text{''}, \qquad\text{where $\abs{u} =  \abs{v} = n/2$}, \]
	then determine and store the message $a_x$ that Alice sends in the BSM SSAT protocol when her input is $\alpha_x$. 
	
	\item For every $y$ with $\abs{y} = n/2$, construct the CNF  
	\[ \beta_y(u, v, z) = \text{``}(y = v) \wedge \phi(u, y, z)\text{''}, \qquad\text{where $\abs{u} = \abs{v} = n/2$}, \]
	then determine and store the message $b_y$ that Bob sends in the BSM SSAT protocol when his input is $\beta_y$. 
\end{enumerate}
 
\noindent{\bf Execution}  On input $xy$ with $\abs{x} = \abs{y} = n/2$, simulate Carol on messages $a_x$ and  $b_y$ and output her answer.

\medskip\noindent
	Under the assumptions of the theorem this algorithm decides $L$ because $\alpha_x \wedge \beta_y$ is satisfiable if and only if $\phi(x, y, \cdot)$ is.  If $\phi(x, y, z)$ accepts then so does $\alpha_x(x, y, z) \wedge \beta_y(x, y, z)$, while if $\alpha_x(u, v, z) \wedge \beta_y(u, v, z)$ accepts it must be that $x = u$ and $y = v$ so $\phi(x, y, z)$ must also accept.
	
	The circuit representation of this algorithm consists of two tables of size $\ell(p(n)) \cdot 2^{n/2}$, a proportional amount of hardware to select the messages $a_x$ and $b_y$, and a copy of Carol's circuit, giving the bound on circuit size.
    \end{proof}

    \begin{corollary} \label{cor:super-hard-NP-language-implies-no-poly-BSM-for-SSAT}
        If there exists a language in $\NP$ that is not computable by any circuit family of size $\wt{O}(2^{n/2})$, then $\SSAT$ doesn't have a $2$-party, $\poly(n)$-size BSM protocol.
    \end{corollary}

\subsection{BSM hardness from worst-case one-way functions} \label{ssec:owf}

We derive consequences of Corollary~\ref{cor:super-hard-NP-language-implies-no-poly-BSM-for-SSAT} to the worst-case insecurity of one-way functions.

A function $f\colon \{0, 1\}^* \to \{0, 1\}^*$ is length-preserving if $\abs{f(x)} = \abs{x}$ for all sufficiently long $x$.  It is non-poly-shrinking if $\abs{f(x)} \geq \abs{x}^\epsilon$ for all sufficiently long $x$ and some constant $\epsilon$.

\begin{proposition}
If every $\NP$ problem has circuit size $\poly(n) \cdot 2^{n/2}$ then every efficient  length-preserving function can be inverted on every output in size $\poly(n) \cdot 2^{2n/3}$.
\end{proposition}
\begin{proof}
Let $L$ be the language
\[ L = \{(y, a)\colon \text{there exists $x$ such that $f(x) = y$ and $a$ is a prefix of $x$}\}. \]
By the assumption $L$ has circuits $C$ of size $\poly(\abs{y} + \abs{a}) \cdot 2^{(\abs{y} + \abs{a})/2}$, which is at most $\poly(n) \cdot 2^{2n/3}$ assuming $\abs{y} = n$ and $\abs{a} \leq n/3$.

Given $y$ in the image of $f$, the circuit $C$ can be applied iteratively to find the first $n/3$ bits of a preimage $x$.  Once those are found, the other $2n/3$ bits can be computed by brute-force search.  The resulting preimage finder has circuit size $\poly(n) \cdot 2^{2n/3}$.
\end{proof}

\begin{proposition}
If every $\NP$ problem has circuit size $\poly(n) \cdot 2^{n/2}$ but $\mathsf{E}^{\NP}_{\parallel}$requires exponential size nonadaptive SAT-oracle circuits then every efficient non-poly-shrinking function can be inverted on every output of length $n$ in size $\poly(n) \cdot 2^{n/2}$.
\end{proposition}
\begin{proof}
Klivans and van Melkebeek~\cite{KvM02} (see also the discussion in~\cite{SU06}) show that under the second assumption there is an efficient deterministic algorithm that on input $y$ of length $n$ produces a list of circuits $V_1, \cdots, V_{\poly(n)}$ such that at least one of the circuits $V_i$ accepts a \emph{unique} $x$ for which $f(x) = y$, assuming $y$ is in the image of the function $f$.  Let 
\[ L = \{(y, i, j)\colon \text{there exists $x$ for which $V_i(x)$ accepts and $x_j = 1$}\}. \]
As $L$ is an $\NP$ language, by the first assumption, it has circuits $C$ of size $\poly(\abs{y} + \abs{i} + \abs{j}) \cdot 2^{(\abs{y} + \abs{i} + \abs{j})/2} = \poly(n) \cdot 2^{n/2}$.  Using $C$, a preimage of $y$ can be found among the rows of the table $x_{ij} = C(y, i, j)$.
\end{proof}

\begin{corollary}
\label{cor:inversion}
If SSAT has a 2-party, $\poly(n)$-size BSM protocol then every efficiently computable function can be inverted on length-$n$ outputs in (1) size $\poly(n) \cdot 2^{2n/3}$ assuming it is length-preserving; (2) size $\poly(n) \cdot 2^{n/2}$ assuming it is non-poly-shrinking and $\mathsf{E}^{\NP}_{\parallel}$requires exponential size nonadaptive SAT-oracle circuits.
\end{corollary}
    
\subsection{Instance hiding and NP languages}
\label{ssec:instance-hiding}

A two-query \emph{instance hiding (IH) scheme}~\cite{BF90}  consists of size circuit (Henry) that receives an input and its output using queries to two oracles (Alice and Bob) such that the distribution of each query depends only on Henry's input length.

The instance hiding scheme is efficient if the size of Henry is polynomial in its input length.  Our main result in this section is the following:

    \begin{theorem} \label{thm:SSAT-SCOL-SPRT-NP-IH}
        If there is an efficient BSM protocol to one of the languages $\SSAT$, $\SCOL$, $\SPRT$, then every language in $\NP$ has an efficient IH scheme.
    \end{theorem}

The connection between IH an BSM will be established using the following notion of reduction.
  
    \begin{definition} \label{def:split-hide-reduction}
        Let $L_1, L_2$ be a languages. A \emph{split-hide reduction from $L_1$ to $L_2$} is a pair of polynomial-time computable randomized mappings $a, b \co \zo^n \to \zo^{\poly(n)}$ satisfying
        \begin{description}
      \item \textit{Correctness:} $x \in L_1 \iff a(x)b(x) \in L_2$.

      \item \textit{Privacy:} The marginal distributions of $a(x)$ and $b(x)$ depend only on the length of $x$. 
        \end{description}
        If such a reduction exists, we say that \emph{$L_1$ split-hide reduces to $L_2$}, and denote $L_1 \leq_p^{\mathsf{sh}} L_2$.
    \end{definition}
    
    \begin{lemma} \label{lem:split-hide-reduction-BSM-to-IH}
        Let $L_1, L_2$ be two languages. If $L_1 \leq_p^{\mathsf{sh}} L_2$ and $L_2$ has an efficient BSM protocol, then $L_1$ has an efficient IH scheme with a single nonadaptive pair of queries.
    \end{lemma}

    \begin{proof}
        Suppose that $L_2$ has a $2$-party efficient BSM protocol implemented by Alice, Bob, and Carol, and that $L_1 \leq_p^{\mathsf{sh}} L_2$ using mappings $a$ and $b$.  In the IH for $L_1$, on input $x$, Henry submits queries $a(x)$ and $b(x)$ to Alice and Bob, respectively, forwards their answers to Carol, and produces her output. 
    \end{proof}

    \begin{corollary} \label{cor:split-hide-reduction-IH-NP}
        Let $L_1, L_2$ be a languages, and suppose that $L_1$ is $\NP$-hard. If $L_1 \leq_p^{\mathsf{sh}} L_2$ and $L_2$ has an efficient BSM protocol, then every language in $\NP$ has an efficient IH scheme with a single nonadaptive pair of queries.
    \end{corollary}

Here, NP-hardness is assumed to hold under reductions that map all instances of a given length $n$ into instances of the same length $m(n)$.

\begin{proof}
In the IH scheme for $L_1$, Henry implements the reduction from $L_1$ to $L_2$ and then runs the IH for $L_2$ from  Lemma~\ref{lem:split-hide-reduction-BSM-to-IH}.
\end{proof}
    
Theorem~\ref{thm:SSAT-SCOL-SPRT-NP-IH} follows from 
Corollary~\ref{cor:split-hide-reduction-IH-NP} and the propositions below which show that $\SAT$, $3\COL$, and $\PARTITION$ each split-hide reduce to their split variant.

    \begin{proposition} \label{prop:SAT-split-hide-reduces-to-SSAT}
        $\SAT \leq_p^{\mathsf{sh}} \SSAT$.
    \end{proposition}

    \begin{proof}
        By the standard reduction from SAT to 3SAT we can represent the input instance by a 3CNF $\varphi(x_1, \dots, x_n)$, where $n$ is chosen large enough to embed all SAT instances of a given size.   We describe the randomized mappings $a$ and $b$ that yield CNFs $\alpha$ and $\beta$, respectively, in variables $x_1, \dots, x_n$ and additional auxiliary variables $y_1, \dotsc, y_{2N}$, where $N = 8\binom{n}{3}$.   Let $C_1,\ldots,C_N$ be a list of all possible clauses of width at most~$3$ over the $n$ variables $x_1,\ldots,x_n$. Let $\pi$ be a random permutation of $[2N]$. We define
        \[
        \alpha = \bigwedge_{i \in [N]} (C_i \lor y_{\pi(i)}).
        \]
        Let $I = \{ i : C_i \in \varphi \}$ and $J = \{ N + i : C_i \notin \varphi \}$. We define
        \[
        \beta = \bigwedge_{i \in I \cup J} \lnot y_{\pi(i)}.
        \]
        \noindent {\it Correctness:} We argue that  $\varphi$ is satisfiable if and only if $\alpha \land \beta$ is satisfiable.
        \begin{itemize}
            \item Suppose that $\varphi$ is satisfiable. Then, we can satisfy $\alpha$ by taking the satisfying assignment of $\varphi$ to satisfy each clause $(C_i \lor y_{\pi(i)})$ for $i \in I$, and assign a true value to $y_{\pi(i)}$ for every $i \in [N] \setminus I$. To satisfy $\beta$, assign a false value to $y_{\pi(i)}$ for every $i \in I$, which will not affect $\alpha$'s satisfiability, and assign a false value to $y_{\pi(i)}$ for every $i \in J$.
            \item Suppose that $\alpha \land \beta$ is satisfiable. Since $\beta$ is satisfiable, it follows that the satisfying assignment assigns false to $y_{\pi(i)}$ for every $i \in I$, which implies that $\alpha$ can be satisfied by the same assignment only if $C_i$ is satisfied for every $i \in I$, namely $\varphi$ itself is satisfied.
        \end{itemize}

        \noindent {\it Privacy:} The $y$-variables in both $\alpha$ and $\beta$ are random length $N$ subsequences of $(y_1, \dots, y_{2N})$.
    \end{proof}

    \begin{proposition} \label{prop:3COL-split-hide-reduces-to-SCOL}
        $3\COL \leq_p^{\mathsf{sh}} \SCOL$.
    \end{proposition}

    \begin{proof}
        Given a graph $G = (V, E)$ with vertices $V = \{1,\ldots,n\}$
        the reduction produces the following two graphs $A$ and $B$ on the same vertex set $\tilde{V}$.  Let $N = \binom{n}{2}$ and $\pi$ be a random permutation of $[2N]$. We identify $[N]$ with ordered pairs of vertices in $V$.
        \begin{itemize}
            \item Vertices $\tilde{V}$: For each $v \in V$ there is a vertex $v$ in $\tilde{V}$, and for each $j \in [2N]$ there are vertices $a_j,b_j,c_j,a'_j,b'_j,c'_j$.
            \item Edges in $A$: For each pair $(v < v') \in V$, let $j = \pi(v, v')$.  We connect $v$ to $a_j,b_j$, $v'$ to $a'_j,b'_j$, and add triangles $(a_j,b_j,c_j)$ and $(a'_j,b'_j,c'_j)$.
            \item The edges in $B$ are
        $\{(c_{\pi(e)},c'_{\pi(e)}) : e \in E \} \cup
        \{(c_{\pi(N + e)},c'_{\pi(N + e)}) : e \notin E \}$.
        \end{itemize}

    \noindent {\it Correctness:}  If $G$ is 3-colorable, the 3-coloring can be extended to a valid coloring of $A \cup B$ by coloring $c_j$ and $c'_j$ for $j = \pi(v, v')$ with the colors of $v$ and $v'$, respectively, and coloring $a_j, b_j, a'_j, b'_j$ with the remaining colors.  The vertices indexed by $j$ that are not in the image of $\pi([N])$ are assigned fixed consistent colors, e.g.,  $(a_j, b_j, c_j, a'_j, b'_j, c'_j) \to (\mathtt{R}, \mathtt{B}, \mathtt{G}, \mathtt{B}, \mathtt{G}, \mathtt{R})$.  Conversely, if $A \cup B$ is 3-colorable, $v$ and $c_j$ must have the same color and so must $v'$ and $c_{j'}$; by the definition of $B$ the coloring on $V$ is valid for $G$.
        
    \medskip\noindent {\it Privacy:}  The graph $A$ consists of $N$ pairs of ``diamonds'' $v, a_j, b_j, c_j$ and $v', a'_j, b'_j, c'_j$, one for each vertex-pair $v < v'$, where $j$ is randomly assigned one of $N$ distinct values in $[2N]$.  The graph $B$ is a matching chosen at random among those that match $N$ out of the $2N$ pairs $(c_1, c'_1), \dots, (c_{2N}, c'_{2N})$.  Both distributions depend on $G$ only through $n$.
    \end{proof}

    \begin{proposition} \label{prop:PARTITION-split-hide-reduces-to-SPRT}
        $\PARTITION \leq_p^{\mathsf{sh}} \SPRT$.
    \end{proposition}
\ignore{ 
    \begin{proof}
        The input to the problem is $n$ non-negative integers of width $n$ bits (by adding zeroes, we can ensure that the number of integers is exactly $n$), forming a set $S$.
We describe (in the following 3 steps) an efficient probabilistic protocol that converts $S$ into two sets of non-negative integers $U_A,U_B$ such that $S$ can be partitioned into two equal parts iff $U_A \cup U_B$ can be partitioned into two equal parts. Moreover, fixing $n$, the marginal distributions of $U_A$ and $U_B$ are independent of $S$. 

        \textbf{Step 1.} Let $p > n2^n$ be prime. \andrej{Choosing this $p$ takes randomness; the reduction will only run in expected poly-time. Does $p$ have to be prime? Can't it be a power of $3$?}
        \dinesh{Agree. Any $p > n2^n$ should work. }
        For any partition $S = A \sqcup B$, if $\sum A \equiv \sum B \pmod{p}$ then $\sum A = \sum B$. This follows from $-n2^n \leq \sum A - \sum B \leq n2^n$.

        \textbf{Step 2.} Denote the original numbers by $x_0,\ldots,x_{n-1}$, and let $y_i,z_i \in \{0,\ldots,p-1\}$ be such that $y_i + z_i \equiv x_i \pmod{p}$. Define
        \begin{align*}
            Y_i &= y_i \cdot 8^n + 1 \cdot 8^i, \\
            Z_i &= z_i \cdot 8^n + 2 \cdot 8^i, \\
            W_i &= 3 \cdot 8^i.
        \end{align*}

        Let $T = \{ Y_i,Z_i,W_i : 0 \leq i \leq n-1 \}$. Suppose that we partition $T = C \sqcup D$ so that $\sum C \equiv \sum D \pmod{8^np}$. Then we can extract a partition of $S$, and vice versa: every partition of $S$ translates to a ``modular partition'' of $T$, as we now show.

        Given a partition $S = A \sqcup B$, we create a similar partition of $T$ as follows. For each $i \in A$, we put $y_i,z_i$ into $C$ and $w_i$ into $D$. Similarly, for each $i \in B$, we put $y_i,z_i$ into $D$ and $w_i$ into $C$. Then $\sum C - \sum D$ is a multiple of $8^n$, and $(\sum C - \sum D)/8^n = \sum A - \sum B$.

        In the other direction, suppose that $\sum C \equiv \sum D \pmod{8^np}$. In particular, $\sum C \equiv \sum D \pmod{8^n}$. Taking $T$ modulo $8$, we get a bunch of zeroes together with $Y_0,Z_0,W_0 = 1,2,3$. The only way to partition these into two sets with equal sum modulo $8$ is by putting $Y_0,Z_0$ into one side and $W_0$ into the other. Continuing in this way, we see that $C,D$ correspond to a partition $S=A \sqcup B$ as in the preceding paragraph, and furthermore, $\sum C \equiv \sum D \pmod{8^np}$ implies $\sum A \equiv \sum B \pmod{p}$.

        \textbf{Step 3.} If $\sum C \equiv \sum D \pmod{8^np}$ then $\sum C - \sum D = \theta p 8^n$, where $|\theta| \leq n$. Let $m = \lceil \log_2 n \rceil$, and let $U$ consist of $T$ together with $U' = \{p8^n/2\} \cup \{ 2^j \cdot p 8^n/2 : 0 \leq j \leq m \}$.\footnote{We could also work with $n$ copies of $p8^n/2$, though this construction is more efficient.}

        On the one hand, $\sum U' = 2^{m+1} \cdot p8^n/2$. On the other hand, the possible sums of subsets of $U'$ are $k \cdot p8^n/2$ for $0 \leq k \leq 2^{m+1}$. Therefore if we partition $U'$ in all possible ways into two sets and compute the sum of the first set minus the sum of the second set, then we get $(2k - 2^{m+1}) \cdot p8^n/2 = (k-2^m) \cdot p8^n$ for $0 \leq k \leq 2^{m+1}$, that is, $k \cdot p8^n$ for $-2^m \leq k \leq 2^m$.

        Given a partition of $U$ into two parts with equal sum, we can extract a partition of $T$ into two parts with equal sum modulo $8^np$, and so a partition of $S$ into two parts with equal sum. Conversely, we can convert a partition of $S$ into two parts with equal sum into a partition of $T$ into two parts with equal sum modulo $8^np$, and by choosing an appropriate partition of $U'$, form a partition of $U$ into two parts with equal sum (since $2^m \geq n$). 

        Concluding, $U$ is a Yes instance of $\PARTITION$ if and only if $S$ is a Yes instance of $\PARTITION$. Note that the size of $U$ is polynomial in the size of $S$.

        \textbf{Construction.} For each $i$, choose $y_i$ at random from $\{0,\ldots,p-1\}$, and let $z_i \in 
\{0,\ldots,p-1\}$ be the unique residue modulo $p$ of $x_i - y_i$. Note that the marginal distributions of $y_0,\ldots,y_{n-1}$ and of $z_0,\ldots,z_{n-1}$ are independent of the instance (given $n$). We give $Y_0,\ldots,Y_{n-1}$ to Alice and the rest of $U$ to Bob.  \andrej{Show the construction first, then the proof of privacy, then the proof of correctness.} \dinesh{Have rewritten it below} \andrej{Thanks! It looks better. I made some small notational changes}
    \end{proof}

\textbf{Re-write of the above proof:}
}
\begin{proof}
	The input to $\PARTITION$ problem consists of a multi-set $S = \{x_1, \dots, x_n\}$ of $n$ non-negative integers (not necessarily distinct) of $n$-bits forming a multi-set $S$.
 	
Reduction: Set $p = n2^n + 1$. For each $i \in \{0,1,\ldots,n-1\}$, pick a $y_i$ uniformly at random from $\{0,1,\ldots,p-1\}$ and let $z_i$ be the unique residue $(x_i-y_i )\mod p$. Define 
	\begin{align*}
		Y_i & = y_i \cdot 8^n + 1 \cdot 8^i \\
		Z_i & = z_i \cdot 8^n + 2\cdot 8^i \\
		W_i & = 3\cdot 8^i		 
	\end{align*}
    Alice's and Bob's inputs are the sets
	\begin{align*}
	 A & = \{Y_0, \ldots, Y_{n-1}\} \\
	B & = \{Z_0, \ldots, Z_{n-1} \} \cup \{W_0, \ldots, W_{n-1}\} \cup E, \qquad E = \{2^j \cdot p8^n \mid 0 \le j \leq \log n \}.
	\end{align*}
		
\medskip\noindent {\it Correctness:}  We show that $U = A \cup B$ has a partition $\sum C = \sum D$ if and only if $S$ has a partition $\sum P = \sum Q$, where $\sum X$ is the sum of all elements in a multi-set $X$.  Let $U' = \{Y_i,Z_i,W_i \mid 0 \le i \le n-1\}$ so that $U' \cup E = U = A \cup B$. We show the following three statements are equivalent.  Correctness follows from the equivalence of 1 and 3.
    \begin{enumerate}
        \item $S$ has a partition $\sum P = \sum Q$
        \item $U'$ has a partition $\sum C' \equiv \sum D' \bmod p8^n$
        \item $U$ has a partition $\sum C = \sum D$.
    \end{enumerate}

    \medskip\noindent 1 $\to$ 2: Assuming $\sum P = \sum Q$, For each $x_i \in P$, place $Y_i$,  $Z_i$ in $C'$ and $W_i$ in $D'$. For each $x_i \in Q$, place $W_i$ in $C'$ and $Y_i, Z_i$ in $D'$. This ensures that $\sum C' - \sum D'$ is a multiple of $8^n$ and moreover $\sum C' - \sum D' = (\sum P - \sum Q)/8^n$. With $P$ and $Q$ of equal sums, $\sum C'-\sum D'\equiv 0 \mod p8^n$, we can conclude that $C',D'$ is the desired partition of $U'$.
	
    \medskip\noindent 2 $\to$ 1:  Assume $\sum C' \equiv \sum D' \bmod p8^n$.   Then for every $i$ (1) $Y_i, Z_i \in C'$ and $W_i \in D'$, or (2) $Y_i, Z_i \in D'$ and $W_i \in C'$, or (3) $Y_i, Z_i, W_i$ do not appear in $C'$ and $D'$.  In all other cases the $i$-th least significant entry in the base-8 representation of $\sum C'$ and $\sum D'$ cannot match.  Assign every number $y_i, z_i, w_i$ to $P$ or $Q$ depending on whether $Y_i, Z_i, W_i$ is in $C'$ or $D'$.  Then $\sum C' \equiv 8^n\sum P + \sum 3 \cdot 8^i$ and $\sum D' \equiv 8^n \sum Q + \sum 3 \cdot 8^i$ modulo $p8^n$, where the second summation is over those indices $i$ 
    that satisfy (1) or (2).
    Therefore $\sum P \equiv \sum Q$ modulo $p$. By our choice of $p$, $\sum P = \sum Q$.  	
    
    \medskip\noindent 2 $\to$ 3:  If $\sum C' \equiv \sum D' \bmod p8^n$, then $\sum C' - \sum D' = k p8^n$ for some $-n \leq k \le n$.  Let $E_k$ be the (unique) subset of $E$ that sums to $\abs{k}p8^n$.  If $k > 0$ set $C = C'$, $D = D' \cup E_k$. Otherwise, set $C = C' \cup E_k$, $D = D'$.  In either case $\sum C = \sum D$.

    \medskip\noindent 3 $\to$ 2: $C'$ and $D'$ are obtained from $C$ and $D$ by dropping the elements in $E$.

\medskip\noindent {\it Privacy:} The marginal distributions of $y_i$ as well as $z_i$ are independent of the input instance $S$ and depends only on $n$. The set $A$ consists of $Y_i$ and $B$  consists of $Z_i, W_i$ and some fixed additional elements none of which depend on the set $S$. Hence, the marginal distributions $A$ and $B$ are also independent of $S$ and depends only on $n$.
\end{proof}

\section{Arithmetic circuits for Matrix Multiplication} \label{sec:arithmetic-circuits-of-bounded-size}

An \emph{arithmetic BSM protocol} for a polynomial $p(X, Y)$ is a decomposition of the form $p(X, Y) = \mathrm{Carol}(\mathrm{Alice}(X), \mathrm{Bob}(Y))$.  The complexity of the protocol is the smallest possible arithmetic circuit complexity of Carol. 

We show that BSM protocols do not help for matrix multiplication.  Recall that the \emph{tensor rank} $R(n)$ of $n \times n$ matrix multiplication is the smallest possible number of terms in a decomposition of $XY$ as a sum of products of linear functions in $X$ and $Y$, respectively.  Asymptotically, $\log_n R(n)$ converges to the matrix multiplication exponent $\omega$.

\begin{theorem}
\label{thm:matrix-multiplication-with-preprocessing}
In any arithmetic BSM protocol for multiplying $n \times n$ matrices Carol must use at least $R(n)/2$ multiplication gates. 
\end{theorem}

As matrix multiplication can be realized in complexity $O(R(n))$ without preprocessing, BSM does not offer any savings in this setting.

\begin{proof}
We prove the contrapositive:  A BSM protocol in which Carol uses $t$ multiplication gates yields a representation of $XY$ of tensor rank at most $2t$.

For a polynomial $P$ in the entries of both $X$ and $Y$, let $c(P)$ be the constant part, $a(P)$ be the linear part involving entries from $X$, $b(P)$ be the linear part involving entries from $Y$, and $\mathit{ab}(P)$ be the bilinear part, where each monomial is a product of an entry of $X$ and an entry of $Y$.  

We construct a circuit that computes the low-degree part $\ell(C) = (c(C), a(C), b(C), ab(C))$ using at most $2t$ multiplication gates inductively over the size of Carol's circuit $C$.  For the base case, Carol has size zero its output must come either from Alice of from Bob, so $ab$ must be zero and $q(C)$ is a linear function of $X$ or $Y$ requiring no multiplications.  For the inductive step, the following rules show that computing $\ell(P + Q)$ from $\ell(P)$ and $\ell(Q)$ takes no extra multiplications, while computing $\ell(PQ)$ takes at most two extra multiplications (underlined):
\begin{center}
\begin{tabular}{l | l}
$c(P+Q) = c(P) + c(Q)$ & $c(PQ) = c(P) c(Q)$ \\
$a(P+Q) = a(P) + a(Q)$ & $a(PQ) = c(P) a(Q) + a(P) c(Q)$ \\
$b(P+Q) = b(P) + b(Q)$ & $b(PQ) = c(P) b(Q) + b(P) c(Q)$ \\
$\mathit{ab}(P+Q) = \mathit{ab}(P) + \mathit{ab}(Q)$ &  $\mathit{ab}(PQ) = c(P) \mathit{ab}(Q) + \mathit{ab}(P) c(Q) + \underline{a(P) \cdot b(Q)} + \underline{b(P) \cdot a(Q)}$
\end{tabular}
\end{center}
As Carol's output is some linear combination of its multiplication gates,  $\ell(XY)$ can be computed using at most twice the number of multiplications used by Carol as desired.
\end{proof}

\section{Polynomials of bounded degree} \label{sec:polynomials-of-bounded-degree}

 Every Boolean function $f \co \zo^N \to \zo$ can be represented by a multilinear  polynomial of degree at most $N$ over the reals. In this section we ask whether we can reduce the degree by means of preprocessing.  Given a function $f \co \zo^{2n} \to \zo$, let $d_f(n, m)$ be the least degree of a polynomial $\mathrm{Carol}$ for which $f(x, y) = \mathrm{Carol}(\mathrm{Alice}(x), \mathrm{Bob}(y))$ and Alice, Bob output at most $m$ bits each.

This type of representation provides dramatic savings for $\AND(x, y) = x_1 \cdot \dotsb \cdot x_n \cdot y_1 \cdot \dotsb \cdot y_n$, which has degree $n$, yet $d_\AND(n, 1) \leq 2$.  This separation is best possible because the only functions with $d_f(n, m) = 1$ are direct sums.

Let $d(n, m) := \max\set{d_f(n, m) \co f \co \zo^{2n} \to \zo}$.  Beaver at al.~\cite{BFKR97} showed (in a different context) that for any constant $C > 0$,
    \[
    d(n,n^{C+1}) \leq \frac{2n}{C \log n}.
    \]
    We extend their argument to obtain the following bound.

    \begin{proposition} \label{prop:universal-degree-reduction-upper-bound}
        For $n < m < 2^n$, it holds that $d(n, m) = O(n/\log(m/n))$.
    \end{proposition}

    \begin{proof}
        Let $f \co \zo^{2n} \to \zo$, and let $p$ be a multilinear polynomial of degree $2n$ that computes $f$:
        \[
        p(z_1, \dots, z_{2n}) = \sum_{S \seq [2n]} c_S \cdot z^S.
        \]

        Consider now the following strategy to reduce the degree of $p$: split the input $z$ into $n/t$ disjoints set of coordinates, each of size $t$, for some chosen $t \leq n$. For each coordinate subset $T \seq [2n]$ of size $T$, we define $2^t - 1$ variable as follows: for each $\emptyset \neq J \seq T$:
        \[
        Z_J = \prod_{i \in J} z_i.
        \]

        It follows that we can express $p$ as a polynomial of degree $\ceil{2n/t}$ by taking each monomial  $z^S$ and replacing it with a product of $Z_J$'s for proper choices of $J$'s.

        We can now establish a BSM protocol for $f$ with Carol computing a degree-$(2n/t)$ polynomial by letting Alice and Bob compute the $Z_J$'s and send the results to Carol. This requires preprocessing output length of at most $\ceil{n/t} \cdot (2^t-1)$. Choosing
        \[
        t = \ceil{\log(m/(4n))} \leq \log(m/(4n)) + 1 = \log(m/n) - 1 \leq \log(2^n/n) < n
        \]
        implies preprocessing of length at most
        \[
        \ceil{n/t} \cdot (2^t-1) \leq (n/t+1) \cdot 2^t \underset{t \leq n}{\leq} (2n/t) \cdot 2^t = (2n/t) \cdot 2^\ceil{\log(m/(4n))} \leq m/t \underset{t \geq 1}{\leq} m.
        \]
        Thus,
        \[
        d(n, m) = \ceil{2n/t} \leq 2n/t+1 = \frac{2n}{\ceil{\log(m/(4n))}} + 1 \leq \frac{2n}{\log(m/(4n))} + 1 = O(n/\log(m/n)). \qedhere
        \]
    \end{proof}

    The following proposition about the equality function shows that this bound is almost tight.

    \begin{proposition} \label{prop:degree-reduction-equality-lower-bound}
        For $n \leq m$, $d_\EQ(n, m) = \Omega(n/\log m)$.
    \end{proposition}

    \begin{proof}
        Let us start by considering the case in which the $\mathbb{F}_2$-polynomial has \emph{individual degree~$1$}, where individual degree-$1$ means that each variable taken from either Alice's message or Bob's message occurs at most once in any monomial computed by Carol. This means that it is of the form
        \[
        P(z,w) = \sum_{ij} c_{ij} z_i w_j + \sum_i d_i z_i + \sum_j e_j w_j + r.
        \]
        Now consider the $n$-bit equality function $\EQ$. Suppose that it can be calculated using preprocessing of length $m$. That is, there exist functions $A,B\colon \zo^n \to \zo^m$ and a polynomial $P$ as above such that for all $x,y \in \zo^n$:
        \[
        \EQ(x,y) = P(A(x),B(y)).
        \]

        For each $x \in \zo^n$, $Q_x(y) = P(A(x),B(y))$ is an affine form on $m$ variables. Any $m+2$ such forms must be linearly dependent. Hence if $m+2 \leq 2^n$, we can find $x_0,\ldots,x_m$ such that
        \[
        Q_{x_0}(y) = Q_{x_1}(y) + \cdots + Q_{x_m}(y).
        \]
        Substituting $y = x_0$, we obtain a contradiction, since the left-hand side should be~$1$, while the right-hand side is a sum of zeroes. We conclude that $m \geq 2^n - 1$. (This should be compared to the obvious upper bound in which $m = 2^n$.)

        Now suppose that $\EQ(x,y) = P(A(x),B(y))$, where $P$ is an arbitrary polynomial of degree~$d$, and $|A(x)| = |B(x)| = m$. Let $A'(x)$ consist of all ANDs of up to $d$ elements from $A(x)$, and define $B'(x)$ similarly. Thus, $|A'(x)| = |B'(x)| = \binom{m}{\leq d}$. We can find a polynomial $P'$ of individual degree~$1$ such that $P'(A'(x),B'(y)) = P(A(x),B(y))$. Therefore,
        \[
        2^n - 1 \leq \binom{m}{\leq d} < (m+1)^d \implies d \geq \frac{n}{\log(m+1)}. \qedhere
        \]
    \end{proof}

    Combining the last two propositions, we have a tight bound when $m$ is polynomial in $n$:

    \begin{corollary} \label{cor:poly-degree-reduction-tight-bound}
        For every constant $C > 1$, $d(n,n^C) = \Theta(n/\log n)$.
    \end{corollary}

\section{Computability with preprocessing} \label{sec:computability-with-preprocessing}

Language $L \subseteq \{0, 1\}^* \times \{0, 1\}^*$ is BSM-computable with message size $m(n)$ if there is a BSM protocol for $L$ in which 
\begin{enumerate}
\item Carol is a Turing Machine that receives $m(\abs{x})$ and $m(\abs{y})$-bit messages from Alice and Bob on input $(x, y)$
\item Carol halts on all inputs of the form $\mathrm{Alice}(x)$, $\mathrm{Bob}(y)$.
\end{enumerate}
In this setting Alice and Bob do not know each other's input length.  

\begin{proposition} \label{prop:computability-with-preprocessing-RE}
Every recursively enumerable language is BSM-computable with message size $2n + 1$.
\end{proposition}

As BSM-computability is closed under complement, the proposition extends to co-recursively enumerable languages.  

\begin{proof}
Let $M$ be a Turing Machine that recognizes $L$, i.e., it halts on yes-instances only.  The following BSM protocol computes $L$:
    \begin{itemize}
        \item On input $x$, Alice sends Carol $x$ and the number of pairs $(x, y')$ with $|y'| \leq |x|$ such that $xy' \in L$. On input $y$, Bob sends Carol $y$ and the number of pairs $(x', y)$ with $|x'| \leq |y|$ such that $x'y \in L$. 
        \item If $\abs{x} \leq \abs{y}$, Carol runs $M$ on all inputs $xy'$ for $|y'| \leq |x|$ in parallel and halts whenever the number of accepted instances  reaches the value sent by Alice. Carol accepts iff $(x, y)$ is among them.  If $\abs{x} <  \abs{y}$, Alice's and Bob's roles are reversed.\hfill \qedhere
    \end{itemize}
\end{proof}

If condition 2 in the definition of BSM-computability were replaced with the stronger requirement that Carol halts on all inputs, Proposition~\ref{prop:computability-with-preprocessing-RE} would no longer be true:  It is possible to construct a recursively enumerable language that is not BSM-computable in this stronger sense whenever $m(n) = o(2^{n/2})$ by diagonalizing against every possible Carol~\cite{O92}.

\section{Conclusion and Open Problems}
In this work, we initiate a systematic study of the BSM complexity of several natural classes of functions with different assumptions on Carol: size-bounded, degree-bounded, and arithmetic (summarized in Table \ref{tbl:result-summary}). Our work suggests several natural open problems.
\begin{itemize}
    \item \textit{Size-bounded Carol.} Can preprocessing help improve the matrix multiplication time when Carol is a Boolean circuit? Our negative result only considered an arithmetic Carol over the reals (Theorem \ref{thm:matrix-multiplication-with-preprocessing}). To argue against the possibility of BSM with polynomial-size Carol for arbitrary  $\NP$ languages, we suggested two explicit candidates for $\NP$ languages with the highest possible circuit complexity: Turing Machine Acceptance and Succinct Subset Sum (see Section \ref{sec:our results introduction} for the exact conjectures). Studying these conjectures and proposing other natural candidates for maximally hard $\NP$ languages may be of independent interest. A final open question is proving an analog of Theorem~\ref{thm:BSM-for-split-SAT-implies-upper-circuit-bound-for-NP} for $\SCOL$, $\SPRT$ or split-versions of some other $\NP$ complete languages, where the natural Cook-Levin approach does not seem to apply. 

    \item \textit{Depth-bounded Carol.} We do not have any results pertaining to a depth-bounded Carol. However, the prior result of \cite{FIKK20} which shows a negative result for a BSM protocol for the mod-$2$ inner product function by $\mathsf{AC}^0$ circuits, leaves a non-trivial open question to solve. Can we prove even a weak unconditional lower bound for the mod-$3$ inner product by $\mathsf{AC}^0$ circuits with mod-$2$ gates?
\end{itemize}

\bibliographystyle{alpha}
\bibliography{biblio}

\newpage
\appendix

\section{Related Models} \label{sec:related-models}

We already saw the relation between bounded simultaneous messages and instance hiding. In this section, we consider two other models that relate to the BSM model---\emph{private information retrieval (PIR)} and \emph{locally decodable codes (LDC)}. The relation is established in the form of reductions between models, which enable the translation of bounds from one model to another. In Section~\ref{ssec:Universal-lower-bound-for-the-instance-hiding-model}, we establish a relation between IH and PIR, and obtain along the way an IH lower bound that applies to most functions; and in Section~\ref{ssec:from-BSM to-LDC-and-vice-versa}, we present the reduction chains relating the four models.

\subsection{Universal lower bound for the instance hiding model} \label{ssec:Universal-lower-bound-for-the-instance-hiding-model}

    Here, we wish to show that efficient instance hiding schemes only capture a sparse set of functions, namely that most functions cannot be computed with an efficient instance hiding scheme. In fact, we show something even stronger, that most functions require Henry to compute a circuit of exponential size.\footnote{In the other direction, any function is computable by an IH scheme if Henry is allowed to compute a circuit of exponential size; indeed, if this is the case, then Henry can compute any function himself without the help of Alice and Bob.}

    To that end, let us consider private information retrieval, introduced in~\cite{CGKS95} and further explored in \cite{Amb96, IK99, BIKR02, BFG06, BIKO12}, mostly in the context of constructing such schemes or ruling out their existence unconditionally. We are interested in the information-theoretic variant defined below.
    \begin{definition}[Private information retrieval] \label{def:PIR}
        A $2$-server, $a$-bit reply PIR protocol for a (binary) database of length $d$ consists of two servers with access to a database of length $d$, and a client that is given as input a database entry of length $n$ and should output its content, while obeying the requirement that the target entry should not be revealed to the servers. The client's queries can be adaptive or non-adaptive.
    \end{definition}

    We also need the corollary of the following known lemma.
    \begin{theorem}[Perles--Sauer--Shelah Lemma~\cite{Sau72, She72}] \label{thm:perles-sauer-shelah-lemma}
        Let $\mathcal{F}$ be a family of subsets of $\zo^n$, satisfying $\size{\mathcal{F}} > \sum_{i=0}^{k-1} \binom{n}{i}$; then, $\mathcal{F}$ shatters a set of size $k$.
    \end{theorem}

    \begin{corollary} \label{cor:perles-sauer-shelah-lemma}
        If the $\VC$ dimension of a set family $\mathcal{F}$ is $k$, then $\size{\mathcal{F}} \leq \sum_{i=0}^k \binom{n}{i} = O(n^k)$.
    \end{corollary}

    Another result we will make use of is the following.
    \begin{theorem}[Kerenidis--Wolf~\cite{KW03}] \label{thm:KW03}
        Any $2$-server, $a$-bit answer PIR protocol for a database of size $N$, requires at least $\Omega(N/2^{5a})$ bits of communication.
    \end{theorem}

    We are now ready to state and prove our claim.
    \begin{proposition} \label{prop:no-IH-for-most}
        For $a \leq n/30$ and $s \leq 2^{n/4}$, most functions on $n$ bits are not computable by $2$-server, $a$-bit answer, $s$-size IH protocols.
    \end{proposition}

    In the proof that follows, we use an idea relating VC dimension and PIR protocols from Beimel \textit{et al.}~\cite[Theorem IV.3]{BIKO12}.

    \begin{proof}
        Let $a \leq n/30$ and $s \leq 2^{n/4}$, and suppose that $\mathcal{F}$ is a function class of functions $\zo^n \to \zo$ such that every function $f \in \mathcal{F}$ admits a $2$-party, $a$-bit answer, $s$-size IH protocol. Since there are at most $\displaystyle 2^{O(s \log s)}$ circuits of size at most $s$, there must be a Henry circuit $C^*$ of size $s$ that computes a function class $\mathcal{F}'$ consisting of $\displaystyle |\mathcal{F}| / 2^{O(s \log s)}$ different functions.

        Now suppose that $d$ satisfies $|\mathcal{F}'| > \binom{N}{<d}$. Viewing $\mathcal{F}'$ as a family of subsets of $\zo^N$, the Perles--Sauer--Shelah Lemma then implies that $\mathcal{F}'$ shatters a set $G \seq \zo^N$ of size $d$. Switching to view $G$ and members of $\mathcal{F}'$ as characteristic vectors in $\zo^N$, this translates to a vector $g$ and a fixed set of $d$ coordinates $I$, such that for any assignment $\alpha \in \zo^d$ to the entries corresponding to $I$, there exists a vector/function $f_\alpha \in \mathcal{F}'$ that agrees with $g$ on $I$.

        Consider now the following $2$-server, $a$-bit answer PIR protocol $\prot{P}'$ for a length-$d$ database:
        \begin{itemize}
            \item \textbf{Input:} Database instance of size $d$.
            \item The content of the database instance is described by some assignment $\alpha \in \zo^d$. The servers will find a function $f_\alpha \in \mathcal{F}'$ from the shattering set corresponding to the database instance.
            \item By assumption, $f_\alpha$ maps to answers returned by the two parties in the promised IH protocol that computes $f_\alpha$. The servers will return the same answers to the client.
            \item The client always runs $C^*$.
        \end{itemize}

        The correctness and privacy requirements of the protocol $\prot{P}'$ follows from the preceding analysis. By Theorem~\ref{thm:KW03}, $\prot{P}'$ requires at least $\Omega(d/2^{5a})$ bits of communication, by which it follows that there exists a database instance requiring that amount of communication, which in turn translates to a function $h$ computed by the IH protocol that uses $C^*$. But if $h$ requires that much of communication to be computed, then $C^*$ must be at least of that size, namely $|C^*| = \Omega(d/2^{5a})$.

        Therefore, if we choose $d$ such that
        \[
        |\mathcal{F}| / 2^{O(s \log s)} \geq 2^{n d} = N^d > \binom{N}{<d},
        \]
        then we get a lower bound of $\Omega(d/2^{5a})$ on the size of the Henry circuit $C^*$. Thus, if we take $d = 2^{n/2}$, we get that any function class $\mathcal{F}$ of size greater than
        \[
        B \teq 2^{n 2^{n/2} + O(s \log s)}
        \]
        must include a function that requires IH protocols that use circuits of size at least
        \[
        \Omega(2^{n/2}/2^{5a}) = \Omega(2^{n/2 - 5a}) \underset{a \leq n/30}{=} \Omega(2^{n/3}).
        \]
        Since $s \leq 2^{n/4}$, we have that $B$ is only but a tiny fraction of $2^{2^n}$, and so most functions cannot be computed by the considered type of IH protocols. Hence, the proposition follows.
    \end{proof}

\subsection{From BSM to locally decodable codes and vice versa} \label{ssec:from-BSM to-LDC-and-vice-versa}

    We now establish a connection, in the form of a two-sided reduction chain, between the BSM model and the following models: instance hiding, private information retrieval, and locally decodable codes, which will be shortly referred to as IH, PIR, and LDC, respectively. The reductions allow us to transfer results (lower and upper bounds) from one model to another. This motivates further study of the BSM model.

    Here, we will consider a slight variant of the BSM model, in which Alice and Bob send Carol their inputs along with some extra bits. In this setting, the output complexity measure is the amount of extra bits instead of the message length. The reason for this choice of model over the original BSM model is that it corresponds better with instance hiding, as we shall see.

    We have already introduced the IH and PIR models, and the following definition introduces LDC (for explicit construction of LDCs, see \cite{KT00, Yek07, DGY10, Efr12}).

    \begin{definition}[Locally decodable codes] \label{def:LDC}
        A \emph{$(q, \delta, \epsilon)$-LDC} is a code $\set{C(x) \co x \in \zo^n} \seq \zo^N$ that satisfies the following property: for every $i \in [n]$, there exists a decoding algorithm $d_i$ that with probability $1 - \epsilon$ can recover $x_i$ from $C(x)$ while only reading $q$ entries, even if (up to) $\delta N$ coordinates of $C(x)$ are corrupted.

        A \emph{$(q,c,\eps)$-smooth code} is a $(q, 0, \eps)$-LDC such that for each $i \in [n]$ and $j \in [N]$, the probability that $d_i$ queries index $j$ is at most $c/N$. A \emph{$q$-perfectly-smooth code} is a $(q, 1, 0)$-smooth code, or equivalently, it is a $(q, 0, 0)$-LDC such that for each $i \in [n]$, every query that $d_i$ makes is (individually) uniform over $[N]$.
    \end{definition}

    Before presenting the reduction chains, let us note that we assume that the IH, PIR, and LDC protocols are non-adaptive. The reductions chain are presented next, and can be illustrated as follows:
    \[
    \boxed{\text{BSM} \iff \text{IH} \iff \text{PIR} \iff \text{LDC}}
    \]

\paragraph*{From BSM to IH.}
    This is Proposition~\ref{prop:bitwise-XOR-BSM-to-IH}: Given $g \co \zo^n \to \zo$, a $2$-party, $a$-extra bits BSM protocol for $g(x+y)$ implies a non-adaptive $2$-party, $a$-bit answer IH protocol for $g$.

\paragraph*{From IH to PIR.}
    The reduction is hidden in the proof of Proposition~\ref{prop:no-IH-for-most}: Let $d$ be an integer, and let $\mathcal{F}$ be a class of functions $f \co \zo^n \to \zo$ such that:
    \begin{itemize}
        \item Each $f \in \mathcal{F}$ admits a $2$-query, $a$-bit answer, $s$-size IH scheme.
        \item $|\mathcal{F}| \geq 2^{O(s \log s)} \cdot \binom{N}{<d}$.
    \end{itemize}
    Then, there exists a $2$-server, $a$-bit reply PIR protocol for database of length $d$.

\paragraph*{From PIR to LDC.}
    We use Lemma 5.1 of Goldriech \textit{et al.}~\cite{GKST02} which shows how to obtain a smooth code from a PIR.
    \begin{theorem}
        Given a $2$-server, $t$-query length, $a$-bit reply  PIR scheme, there is a $(2, 3, 0)$-smooth code $C:\zo^n \to \Sigma^m$ for $\Sigma= \zo^a$ and $m \le 6\cdot 2^t$.
    \end{theorem}

\paragraph*{From IH to BSM.}
    Suppose there is a (non-adaptive) $2$-query, $a$-bit answer, $s$-size IH protocol for a function $g \co \zo^n \to \zo$, and suppose that both queries made by Henry are uniformly distributed over $\zo^m$. Then, there is a $2$-party, $a$-extra bits, $s$-size BSM protocol for the function $f \co \zo^m \times \zo^m \to \zo$ defined by $f(x, y)$ equals $g(z)$ where $z$ yields $x, y$ as queries in the IH protocol (note that this is well-defined).

\paragraph*{From PIR to IH.}
    It is easy to see that a $2$-server, $a$-bit reply PIR protocol for a database of length $d$, implies a $2$-query, $a$-bit answer IH protocol for every function $f \co \zo^{\floor{\log d}} \to \zo$, where the size of Henry is bounded by the size of the PIR's client.

\paragraph*{From LDC to PIR.}
    We can easily convert a $q$-perfectly-smooth code to a $q$-server, $1$-bit reply PIR protocol for database of length $n$: given a database content $x \in \zo^n$ and an entry $i \in [n]$, the client will apply the encoder $d_i$, and each server will encode $x$ to get $C(x)$ and then return the bit asked by the LDC query.

\section{Implicit BSM Bounds} \label{sec:implicit-BSM-bounds}

Here, we present results that follows immediately or are easily obtainable from known results in the literature concerning related models.

\subsection{Bitwise-Combined Functions} \label{ssec:bitwise-combined-functions}

    A question we can ask is what about BSM protocols whose target is limited to functions of the form $f(x, y) = g(x * y)$, where $*$ is a concrete combiner function, such as bitwise XOR/OR/AND? Let us call such functions $f$ as \emph{bitwise-$*$ combined} functions (also known as \emph{lifted functions}). Combiners play a role in connecting bounded simultaneous messages with instance hiding.

    We can tackle the problem of finding lower bounds for combined functions in the bitwise-XOR case using the following proposition, which establishes a relation between IH protocols and BSM protocols for bitwise-XOR combined target functions.

    \begin{proposition} \label{prop:bitwise-XOR-BSM-to-IH}
        Let $g \co \zo^n \to \zo$ be a Boolean function, and let $f \co \zo^n \times \zo^n \to \zo$ be defined by $f(x, y) = g(x + y)$. Suppose that $f$ has a $2$-party, $a$-extra bits BSM protocol. Then, $g$ has a $2$-party, $a$-bit answer IH protocol. The complexity of the IH's Henry is the same complexity of the BSM's Carol plus the complexity of secret sharing the input with a random string (a trivial operation).
    \end{proposition}

    \begin{proof}
        The reduction is simple: given a $2$-party, $a$-extra bits BSM protocol, to form an IH protocol, have Henry secret share his input between Alice and Bob, which will simulate their BSM counterparts; correctness follows from the relation $g(x) = g((x + r) + r) = f(x + r, r)$.
    \end{proof}

    This relation to IH protocols might motivate further study of BSM protocols.

\subsubsection{From BSM protocols to circuit upper bounds} \label{sssec:from-BSM-protocols-to-circuit-upper-bounds}

    Using a similar technique to the one used in the proof of Proposition~\ref{prop:lower-bound-for-most-bitwise-XOR combined-functions}, we can obtain a nontrivial upper bound on the circuit size required to compute a function whose bitwise-XOR combined version has a BSM protocol.

    \begin{proposition} \label{prop:BSM-to-circuit-upper-bound}
        Let $g \co \zo^n \to \zo$ be a function, and suppose there exists a $2$-party, $s$-size BSM protocol for $f(x, y) = g(x + y)$. Then, there is a circuit of size $O((s/n) \cdot 2^{n/2})$ computing $g$.
    \end{proposition}

    \begin{proof}
        Let $\prot{P} = (A, B, C)$ be a $2$-party, $s$-size BSM protocol for $g$. Consider the following BSM protocol $\prot{P}' = (A', B', C')$ that operates on $n$ bit inputs, denoted $x = (x_L, x_R)$ with $x_L, x_R \in \zo^n$:
        \begin{itemize}
            \item $A'(x_L) = A(x_L, 0^{n/2})$ and $B'(x_R) = B(0^{n/2}, x_R)$.
            \item $C' = C$.
        \end{itemize}
        By the correctness of $\prot{P}$, we have:
        \[
        C'(A'(x_L), B'(x_R)) = C(A(x_L, 0^{n/2}), B(0^{n/2}, x_R)) = g((x_L, 0^{n/2}) + (0^{n/2}, x_R)) = g(x_L, x_R),
        \]
        which implies that $\prot{P}'$ computes $g$.

        Since the circuit size of $C$ is $s$, we can assume without loss of generality that both $A'$ and $B'$ are $\zo^{n/2} \to \zo^{s}$ functions. This means that we can compute both functions using two brute--force circuits each of size $s \cdot O(2^{n/2}/n)$. Concatenating these circuits with the circuit $C'$ computes, we get a circuit for $g$ of size $O((s/n) \cdot 2^{n/2})$.
    \end{proof}

    Let us note that the same result can be obtained in a similar way for bitwise-OR and bitwise-AND combiners. Also note that the contraposition of Proposition~\ref{prop:BSM-to-circuit-upper-bound} translates circuit lower bounds to BSM lower bounds (for bitwise-XOR/OR/AND combined functions).

\subsubsection{Combiners and structural complexity} \label{sssec:combiners-and-structural-complexity}

    Let $L$ be a language, and consider the following bitwise-XOR combined version of $L$:
    \[
    L_\oplus := \set{(x,y) \co x \oplus y \in L}.
    \]
    In what follows, we translate BSM protocols for such combined $\NP$-hard and $\PSPACE$-hard languages into implications in structural complexity.

    To obtain the first implication, we will need to rely on the following two theorems.
    \begin{theorem}[Fortnow--Szegedy~\cite{FS91}] \label{thm:FS91}
        If a language $L$ has a $2$-server, $1$-bit answer, $\poly(n)$-size IH protocol, then $L \in \NP/\poly \cap \coNP/\poly$.
    \end{theorem}

    \begin{theorem}[Yap~\cite{Yap83}] \label{thm:Yap83}
        If $\NP \seq \coNP/\poly$, then the polynomial hierarchy collapses to the third level.
    \end{theorem}

    We can now state the following corollary.
    \begin{corollary} \label{cor:BSM-NP-hard-HP-collapse}
        Suppose that $L$ is $\NP$-hard, and that $L_\oplus$ has a $2$-party, $1$-extra bit, $\poly(n)$-size BSM protocol. Then, the polynomial hierarchy collapses to the third level.
    \end{corollary}

    \begin{proof}
        If $L_\oplus$ has a $2$-party, $1$-extra bit, $\poly(n)$-size BSM protocol, then, by Proposition~\ref{prop:bitwise-XOR-BSM-to-IH}, $L$ has a $2$-party, $1$-bit, $\poly(n)$-size answer IH protocol. If $L$ is $\NP$-hard, then, by Theorem~\ref{thm:FS91}, it follows that $L \in \NP/\poly \cap \coNP/\poly$. Theorem~\ref{thm:Yap83} completes the proof.
    \end{proof}

    Next, consider now the following result.
    \begin{theorem}[Corollary 4.4 in \cite{Tri10}] \label{thm:Tri10}
        If a language $L$ has a $2$-server, $O(\log n)$-bit answer, $\poly(n)$-size IH protocol, then $L \in \PP^\NP/\poly$.
    \end{theorem}

    Now, combining Theorem~\ref{thm:Tri10} and Proposition~\ref{prop:bitwise-XOR-BSM-to-IH}, we get the following corollary.

    \begin{corollary} \label{cor:BSM-PSPACE-hard-PP-NP/poly}
        Suppose that $L$ is $\PSPACE$-hard, and that $L_\oplus$ has a $2$-party, $O(\log n)$-extra bit, $\poly(n)$-size BSM protocol. Then, $\PSPACE \seq \PP^\NP/\poly$.
    \end{corollary}

\subsubsection{Universal lower bound for bitwise-combined functions} \label{sssec:universal-lower-bound-for-bitwise-combined-functions}

    We show now that most combined functions are hard to compute in the BSM model.
    \begin{proposition} \label{prop:lower-bound-for-most-bitwise-XOR combined-functions}
        Most bitwise-XOR-combined functions $f \co \zo^n \times \zo^n \to \zo$ require BSM protocols with preprocessing length of $\tilde{\Omega}(2^{n/2})$.
    \end{proposition}

    \begin{proof}
        Let $g \co \zo^n \to \zo$ be a function, let $f \co \zo^n \times \zo^n \to \zo$ be defined by $f(x, y) = g(x + y)$, and suppose there is a BSM protocol $\prot{P} = (A, B, C)$ computing $f$ with preprocessing output length equals to some positive integer $\ell$.

        For $z \in \zo^n$, let us denote $z = (z_L, z_R)$, where $z_L, z_R \in \zo^{n/2}$.\footnote{Since we are after lower bounds, we can safely assume that $n$ is even.} The correctness of $\prot{P}$ implies that
        \[
        g(z) = g( (z_L, 0^{n/2}) + (0^{n/2}, z_R) ) = f( (z_L, 0^{n/2}) , (0^{n/2}, z_R) ) = C( A(z_L, 0^{n/2}) , B(0^{n/2}, z_R) ).
        \]
        Thus, we can deduce that every $n$-bit function computable by a BSM protocol with preprocessing output length of $\ell$, corresponds to a triplet consisting of two $\zo^{n/2} \to \zo^\ell$ functions and a $2\ell$-bit circuit. The number of functions that can be described by such triplets is bounded by:
        \[
        (2\ell)^{O(2\ell)} \cdot (2^\ell)^{2^{n/2}} \cdot (2^\ell)^{2^{n/2}} = 2^{O(2\ell) \cdot \log(2\ell) + \ell \cdot 2^{n/2+1}} = 2^{O(\ell \cdot 2^{n/2})}.
        \]
        Since there are $2^{2^n}$ functions on $n$ bits, the proposition follows.
    \end{proof}

    We can apply the same argument and obtain the same lower bound for bitwise-OR and bitwise-AND combined functions using the following identities:
    \begin{align*}
    z &= (z_L, 0^{n/2}) \vee (0^{n/2}, z_R), \\
    z &= (z_L, 1^{n/2}) \wedge (1^{n/2}, z_R).
    \end{align*}
    
    Note that we give here a direct counting argument; however, it is also possible to obtain an exponential lower bound, yet with worse parameters, using Proposition~\ref{prop:no-IH-for-most} and the reduction given in Proposition~\ref{prop:bitwise-XOR-BSM-to-IH}.

\subsubsection{Universal upper bound for bitwise-combined functions} \label{sssec:universal-upper-bound-for-bitwise-combined-functions}

    We present a nontrivial upper bound on the extra bits required to compute a bitwise-XOR combined function. To that end, we have the following definition.

    \begin{definition}[Generalized Addressing Function (GAF)] \label{def:GAF}
        Let $(G, +)$ be a group of order $n$, and let $1 \leq k$ be an integer. The \emph{Generalized Addressing Function} for $G$ and $k$,
        \[
        \func{GAF}_{G, k} \co \zo^{n + (k-1)\log n} \to \zo,
        \]
        is defined by
        \[
        \func{GAF}_{G, k}(x_0, x_1, \dotsc, x_{k-1}) = x_0(x_1 + \dotsb + x_{k-1}),
        \]
        where $x_1, \dotsc, x_{k-1} \in G$ are represented by binary strings of length $\log n$, and $x_0 \co G \to \zo$ is represented as an $n$-bit string.
    \end{definition}

    In~\cite{BGKL03}, they define the SM complexity of a function $f$, denoted by $C_0(f)$, to be the minimum cost of an SM protocol computing $f$, where the cost of a protocol is defined as the longest message sent to the referee by any individual player. In the same paper, they obtain a nontrivial (and surprising) upper bound of $O(2^{0.92 n})$ on the SM complexity of $\func{GAF}_{\Z_2^n, 3}$. This upper bound is further improved to $O(2^{0.729 n})$ in ~\cite{AL00}.

    \begin{theorem}[Ambainis--Lokam~\cite{AL00}] \label{thm:AL00}
        There exists an SM protocol for $\func{GAF}_{\Z_2^n, 3}(g, x, y)$ with SM complexity of $O(2^{(0.728... + o(1)) n})$.
    \end{theorem}

    To make use of this result, we need to following proposition, which translates SM upper bounds for the generalized addressing functions into BSM upper bounds on the communication required to compute any bitwise-XOR combined function.

    \begin{proposition} \label{prop:SM-for-GAF-implies-BSM-for-all}
        Denote $s \teq C_0(\func{GAF}_{\Z_2^n, 3})$. Then, for every bitwise-XOR combined function $f(x, y) = g(x + y)$, there exists a  $2$-party, $O(s)$-communication BSM protocol that computes $f$.
    \end{proposition}

    \begin{proof}
        First, let us note that the SM model considered in~\cite{AL00} is a variant in which the $i$th input is written on the forehead of player $i$ (aka number-on-the-forehead model), while the BSM variant we consider is one in which each player sees its own input (aka number-in-hand model); however, these two variants coincide when there are two players.

        Let $\prot{P} = (A, B, C)$ be an SM protocol for $\func{GAF}_{\Z_2^n, 3}(g, x, y)$ that achieves the $s$ upper bound, and let $g \co \zo^n \to \zo$ be an arbitrary function. We construct a BSM protocol $\prot{P}_g = (A_g, B_g, C_g)$ for $f(x, y) = g(x + y)$ as follows:
        \begin{itemize}
            \item Let us denote by $\alpha(x)$ and $\beta(y)$ the messages sent to Carol of $\prot{P}$ by Alice and Bob of $\prot{P}$, respectively.
            \item Alice and Bob of $\prot{P}_g$ will send $(x, \alpha(x))$ and $(y, \beta(y))$, respectively, to Carol of $\prot{P}_g$.
            \item Since Carol of $\prot{P}$ has no access to $g$, $C$ does not depend on $g$ but it can depend on $x$ and $y$ directly, so in general it is a function of $(x, y, \alpha(x), \beta(y))$, which Carol of $\prot{P}_g$ can simulate.
        \end{itemize}

        By the correctness and complexity of the assumed protocol $\prot{P}_g$, the protocol defined above computes $f$ correctly with the promised message length.
    \end{proof}

    \begin{corollary}[BSM upper bound] \label{cor:BSM-upper-bound}
        For every bitwise-XOR combined function $f(x, y) = g(x + y)$, there exists a  $2$-party, $O(2^{(0.728... + o(1)) n})$-extra bits BSM protocol that computes $f$.
    \end{corollary}

    \begin{proof}
        Follows from Theorem~\ref{thm:AL00} and Proposition~\ref{prop:SM-for-GAF-implies-BSM-for-all}.
    \end{proof}

    The following proposition is a converse to Proposition~\ref{prop:SM-for-GAF-implies-BSM-for-all}.

    \begin{proposition} \label{prop:BSM-for-all-implies-SM-for-GAF}
        Suppose that for every bitwise-XOR combined function $f(x, y) = g(x + y)$, there exists a  $2$-party, $s$-size BSM protocol that computes $f$. Then, $C_0(\func{GAF}_{\Z_2^n, 3}) = O(s \log s) = \tilde{O}(s)$; furthermore, the circuit size required for the referee's computation is bounded by $O(s \log^2 s)$.
    \end{proposition}

    \begin{proof}
        Suppose the antecedent of the proposition holds. For every function $g$, let us denote by $\prot{P}_g = (A_g, B_g, C_g)$ the $2$-party, $s$-size BSM protocol that computes $f(x, y) = g(x + y)$. As in the proof of Proposition~\ref{prop:SM-for-GAF-implies-BSM-for-all}, in what follows we rely on the fact that $f(x, y) = f(y, x)$.

        We construct an SM protocol $\prot{P} = (A, B, C)$ for $\func{GAF}_{\Z_2^n, 3}(g, x, y)$ as follows:
        \begin{itemize}
            \item Given input $(g, x, y)$, $A$ and $B$ will output whatever $A_g$ and $B_g$ output on $(x, y)$; in addition, $A$ will output a representation of the circuit given by $C_g$.
            \item On input $\bigl((C_g, A_g(x)), B_g(y)\bigr)$, $C$ will evaluate $C_g$ on input $(A_g(x), B_g(y))$ and output the result.
        \end{itemize}

        The correctness of $\prot{P}$ follows from that of $\prot{P}_g$. Since the size of $C_g$ is bounded by $s$, so are the messages of $A_g$ and $B_g$, and furthermore, we can represent circuits of size up to $s$ using $O(s \log s)$ bits. Thus, the communication cost of $A$ and $B$ is bounded by $O(s \log s)$, by which the first half of the consequent of the proposition follows. The second half of the consequent, namely the bound on the size of $C$, follows from the circuit complexity of circuit evaluation.
    \end{proof}

\paragraph{Upper bounds via the GAF function for OR/AND}
    All the upper bounds on the SM complexity of $\func{GAF}_{\Z_2^n, 3}$ obtained in~\cite{BGKL03} and~\cite{AL00} rely on a key lemma stating the following: If a function $g \co \zo^n \to \zo$ is representable by a multilinear polynomial of degree $d$ over $\Z_2$, then there exists an SM protocol for $g$ with communication cost of $\binom{n}{\leq \floor{d/2}}$. The proof of this lemma takes the multilinear polynomial representation of $g$,
    \[
    g(x + y) = \sum_{|S| \leq d} g_S \prod_{i \in S} (x_i + y_i),
    \]
    and rearranges it into a different polynomial representation of $g$ of smaller degree, such that each coefficient of this polynomial depends either solely on Alice's input or on Bob's, and this allows for fewer bits to be communicated. Unfortunately, if we replace $x_i + y_i$ with $x_i \wedge y_i$, it's not clear if we can manipulate the polynomial representation in the same way, so we might look into a different strategy to tackle the OR/AND bitwise combiner. One natural possibility is to consider the DNF/CNF representation of a function $g$. This enables us to obtain the following simple claim.

    \begin{proposition} \label{prop:short-DNF-implies-nontrivial-BSM-protocol}
        If a function $g$ is representable by a DNF with $t$ terms, then $f(x, y) = g(x \wedge y)$ has a $2$-party, $t$-size BSM protocol.
    \end{proposition}

    \begin{proof}
        Consider a term of $g$, which has the following general form: $z_{i_1} \wedge \dotsb \wedge z_{i_k}$. Substituting $z_{i_j}$ with $x_{i_j} \wedge y_{i_j}$, we get $\left(\bigwedge_{i = 1}^k x_{i_j}\right) \wedge \left(\bigwedge_{i = 1}^k y_{i_j}\right)$. This leads to a BSM protocol in which Alice and Bob compute their respective parts for each term, and pass the results to Carol who computes their OR. This requires $t$ bits of communication by each party and a circuit of size $t$ for Carol.
    \end{proof}

    \begin{corollary}\label{cor:monotone-k-DNF-implies-efficient-BSM-protocol}
        If a function $g$ is representable by a monotone $k$-DNF, then $f(x, y) = g(x \wedge y)$ has a $2$-party, $O(n^k)$-size BSM protocol.
    \end{corollary}

    \begin{corollary}\label{cor:decision-tree-depth-implies-BSM-protocol}
        Every function $g$ admits a $2$-party, $O\left(2^{D(g)}\right)$-size BSM protocol, where $D(g)$ is the depth of the shallowest decision tree computing $g$.
    \end{corollary}

    Note that we can obtain similar results for the bitwise-OR combiner if we replace DNF with CNF.

    Consider now the special case of $g$ being an OR function. Here is a challenge: Computing $g(x \wedge y)$ without preprocessing requires a circuit of size at least $2 n$ (or otherwise it won't read all inputs). Can we reduce the circuit complexity if we allow separate preprocessing of $x$ and $y$? A positive answer will imply, recalling the proof of Proposition~\ref{prop:short-DNF-implies-nontrivial-BSM-protocol},  a nontrivial BSM protocol for $g(x \wedge y)$, for all functions $g$. A negative answer might hint why it is difficult to upper bound the general case, as most functions have many terms in their shortest DNF representation.

\paragraph{Monotone functions}
    It turns out that we can do better if we know that $g$ is monotone. Let us consider the following general monotone DNF representation of a monotone function $g \co \zo^n \to \zo$:
    \[
    g(z_1, \dotsc, z_n) = \bigvee_{S \seq [n]} \left( g_S \wedge \bigwedge_{i \in S} z_i \right),
    \]
    where the $g_S$ are constants in $\set{T, F}$. Let us also denote $Z_{\wedge S} \teq \bigwedge_{i \in S} z_i$, with $Z_{\wedge \emptyset} \teq T$.

    We will say that $g$ has \emph{DNF-width} $w$ if the narrowest monotone DNF representing $g$ is of width $w$. Note that in such a case, $g_S = F$ for every $S$ with $|S| > w$, which means we can discard terms wider than $w$ from the representation. Inspired by Lemma 5.1 of~\cite{BGKL03}, we obtain the following analogous lemma for bitwise-OR combiners.

    \begin{lemma} \label{lem:narrow-monotone-implies-efficient-BSM-protocol}
        If a monotone function $g \co \zo^n \to \zo$ has DNF-width $w$, then $g(x \vee y)$ admits a $2$-party, $O(\binom{n}{\leq \floor{w/2}})$-size BSM protocol.
    \end{lemma}

    \begin{proof}
        We have:
        \begin{gather*}
        g(x \vee y)
        = \bigvee_{|S| \leq w} \left( g_S \wedge \bigwedge_{i \in S} (x_i \vee y_i) \right)
        = \bigvee_{|S| \leq w} \left( g_S \wedge \bigvee_{T_1 \cup T_2 = S} \left( X_{\wedge T_1} \wedge Y_{\wedge T_2} \right) \right) \\
        = \left[ \bigvee_{|T_1| \leq \floor{w / 2}} \left( \bigvee_{|T_1 \cup T_2| \leq w} g_{T_1 \cup T_2} \wedge Y_{\wedge T_2} \right) \wedge X_{\wedge T_1} \right]
            \vee \left[ \bigvee_{|T_2| \leq \floor{w / 2}} \left( \bigvee_{|T_1 \cup T_2| \leq w} g_{T_1 \cup T_2} \wedge X_{\wedge T_1} \right) \wedge Y_{\wedge T_2} \right] \\
        = \left( \bigvee_{|T_1| \leq \floor{w / 2}} \Phi_{T_1}(y) \wedge X_{\wedge T_1} \right)
            \vee \left( \bigvee_{|T_2| \leq \floor{w / 2}} \Phi_{T_2}(x) \wedge Y_{\wedge T_2} \right),
        \end{gather*}
        where $T_1$ and $T_2$ are disjoint, $\Phi_{T_1}(y)$ replaces clauses that depend only on $y$, and $\Phi_{T_2}(x)$ replaces clauses that depend only on $x$.

        This gives way to the following BSM protocol: Alice can send Carol the values of the $\Phi_{T_2}(x)$ and $X_{\wedge T_1}$ terms, Bob can send Carol the values of the $\Phi_{T_1}(y)$ and $Y_{\wedge T_2}$ terms, and Carol can use a DNF of size $O(\binom{n}{\leq \floor{w/2}})$ to complete the computation of the formula. This completes the proof.
    \end{proof}

    We can now use Lemma~\ref{lem:narrow-monotone-implies-efficient-BSM-protocol} to prove an analogous proposition to Theorem 5.3 in~\cite{BGKL03}. The proof of our proposition is pretty much the same as the proof of their theorem, except that instead of the polynomial representation and degree of a function we manipulate its DNF representation and width.

    \begin{proposition} \label{prop:monotone-implies-nontrivial-BSM-protocol}
        If $g \co \zo^n \to \zo$ is monotone, then $g(x \vee y)$ admits a $2$-party, $o(2^{0.92 n})$-size BSM protocol.
    \end{proposition}

    \begin{proof}
        We have:
        \[
        g(z) = \left( \bigvee_{|S| \leq 2n/3} g_S \wedge Z_{\wedge S} \right) \vee \left( \bigvee_{|S| > 2n/3} g_S \wedge Z_{\wedge S} \right).
        \]
        To compute the left disjunction, we can use the protocol provided in the proof of Lemma ~\ref{lem:narrow-monotone-implies-efficient-BSM-protocol}. As for the right disjunction, Alice and Bob can send Carol directly the high-width coefficients, namely $X_{\wedge S}, Y_{\wedge S}$ for every $S$ with $|S| > 2n/3$, together with their inputs $x$ and $y$, and Carol can compute the disjunction with a circuit of size $O(n \cdot \binom{n}{> 2n/3})$.\footnote{The $n$ multiplicand is because Carol needs to compute $\bigwedge_{i \in S} (x_i \vee y_i)$ for every $S$ with $|S| > 2n/3$.} Finally, an additional OR gate is required to compute the whole formula. The overall size of the required circuit is:
        \[
        O\left(\binom{n}{\leq n/3}\right) + O\left(n \cdot \binom{n}{> 2n/3}\right) = O\left(n \cdot \binom{n}{\leq n/3}\right) = O(n \cdot 2^{H(1/3) n}),
        \]
        where $H$ is the binary entropy function, and we used the inequality $\binom{n}{\leq k} \leq 2^{H(k/n) n}$ that holds when $k \leq n/2$. Since $H(1/3) \approx 0.918 < 0.92$, the proposition follows.
    \end{proof}

    \begin{remark}
        \mbox{}
        \begin{itemize}
            \item For the bitwise-AND combiner, we have analogous lemma and proposition to Lemma~\ref{lem:narrow-monotone-implies-efficient-BSM-protocol} and Proposition~\ref{prop:monotone-implies-nontrivial-BSM-protocol} if we consider the CNF representation of monotone functions instead of the DNF representation.
            \item Using Lemma~\ref{lem:narrow-monotone-implies-efficient-BSM-protocol} and the argument in the proof of Proposition~\ref{prop:monotone-implies-nontrivial-BSM-protocol}, we can obtain a similar SM upper bound as in Theorem 5.3 of~\cite{BGKL03}, but for a variant of the general addressing function, where bitwise-XOR is replaced with bitwise-OR/AND, and we are promised that the input function is monotone.
            \item Note that Lemma~\ref{lem:narrow-monotone-implies-efficient-BSM-protocol} can be easily extended to include all unate functions, namely functions that have a DNF representation in which every variable appears consistently either as a positive literal or negative literal across the entire formula; and formally, functions $g$ for which there exists a monotone function $h$ and a subset $S \seq [n]$ such that $g(z + S) = h(z)$ for all $z \in \zo^n$.
        \end{itemize}
    \end{remark}

    We extend the above result to obtain a BSM protocol for any Boolean function where Carol is a constant depth circuit of size $o(2^{0.92 n})$. This uses a structural characterization of Boolean functions with bounded alternation due to Blais \textit{et al.}~\cite{BCOST15}, which in turn is an extension of a classical result due to Markov~\cite{Mar58}.

    We first describe what is alternation of a Boolean function. Viewing the $n$ bit Boolean hypercube as a poset and for $0 \le k \le n$, a Boolean function $f$  on $n$ bit is said to have \textit{alternation} as $k$, if for any chain in the hypercube from $0^n$ to $1^n$, the number of times the function value changes along the strings in the chain is at most $k$. Observe that alternation is always bounded by $n$ since any chain in an $n$ bit hypercube can be of length at most $n$. It can be observed that monotone functions have an alternation of $1$.

    \begin{theorem}\label{thm:monotone-implies-nontrivial-BSM-protocol-via-alternation}
        For any $g \co \zo^n \to \zo$, $g(x \vee y)$ admits a $2$-party, $o(2^{0.92 n})$-size BSM protocol with a Carol of depth $3$.
    \end{theorem}

    \begin{proof}
        Let the alternation of $g$ be $k$. We start with a result of~\cite{BCOST15} who showed that a Boolean function $g$ has alternation $k$ if and only if there exists $k$ monotone functions $g_i \co \zo^n \to \zo$ such that $g(x) = b \oplus g_1(x) \oplus \cdots \oplus g_k(x)$, where $b:= g(0^n)$. 

        Applying Proposition~\ref{prop:monotone-implies-nontrivial-BSM-protocol} to each of the $g_i$ gives a $2$-party, $o(2^{0.92n})$-size BSM protocol. The final protocol is to compute  $g_i(x \lor y)$ for each $i \in \{1,2\ldots,k\}$. For each $g_i$, Carol gets the final value after doing an OR of the bits received. Following this, Carol uses a depth $3$ circuit to compute parity of these $k$ values as the final output. By merging the gates at the bottom level, the depth continues to be $3$. Correctness follows by the above characterization. With $k \le n$, the overall communication will be $o(k2^{0.918n}) = o(2^{0.92n})$.
    \end{proof}

    Observe that the above result does not match the best bounds for XOR combiners due to Ambainis and Lokam~\cite{AL00}. In what follows, we show that similar to the upper bound improvements provided in~\cite{AL00}, we not only improve the upper bound of Proposition~\ref{prop:monotone-implies-nontrivial-BSM-protocol}, but this time include all functions, and not only monotone.

    \begin{lemma} \label{lem:decomposition-into-bounded-width-monotone-DNFs}
        Let $g \co \zo^n \to \zo$ be a function. Then, for each $0 \leq i \leq n$, there is a unate function $g_i$ with DNF-width of at most $n/2$ that coincides with $g$ on inputs of weight $i$.
    \end{lemma}

    \begin{proof}
        Let us define the following function for every $w \in \zo^n$:
        \[
        \Phi_w(z) :=
        \begin{cases} 
        \bigwedge_{w_i = 1} z_i & \text{if } |w| \leq n/2, \\
        \bigwedge_{w_i = 0} \overline{z_i} & \text{if } |w| > n/2.
        \end{cases}
        \]

        We can now define $g_i$ as follows:
        \[
        g_i(z) = \bigvee_{|w| = i} g(w) \wedge \Phi_w(z).
        \]

        Clearly, $g_i$ is unate with DNF-width of at most $n/2$. Now if $|z| = i$, then for every $w$ with $|w| = i$ we have that $\Phi_w(z)$ evaluates to true if and only if $w = z$; thus, $g_i(z) = g(z)$ for every $z$ with $|z| = i$, as desired.
    \end{proof}

    \begin{theorem} \label{thm:BSM-upper-bound-for-all-OR-AND-combiners}
        Let $g \co \zo^n \to \zo$ be any function. Then, $g(x \vee y)$ admits a $2$-party, $o(2^{0.82 n})$-size BSM protocol.
    \end{theorem}

    \begin{proof}
        Let $g \co \zo^n \to \zo$ be a function. Consider the following BSM protocol:
        \begin{itemize}
            \item For every $0 \leq i \leq n$, Alice and Bob will construct $g_i$, as described in Lemma~\ref{lem:decomposition-into-bounded-width-monotone-DNFs}, and simulate the protocol of Lemma~\ref{lem:narrow-monotone-implies-efficient-BSM-protocol}. In addition, Alice and Bob will pass their inputs to Carol.
            \item Carol will compute the weight of $x \vee y$ to decide which protocol, out of the $n + 1$ protocols, it should follow to output the correct answer.
        \end{itemize}

        Correctness follows from Lemma~\ref{lem:decomposition-into-bounded-width-monotone-DNFs}, and the circuit size required, by Lemma~\ref{lem:narrow-monotone-implies-efficient-BSM-protocol}, is
        \[
        O\left(n \cdot \binom{n}{\leq n/4}\right) = O(n \cdot 2^{H(1/4) n}),
        \]
        plus an additional negligible polynomial because of the circuit Carol needs in order to compute the weight of $x \vee y$. Since $H(1/4) \approx 0.811 < 0.82$, the theorem follows.
    \end{proof}

    It turns out that we can do even better using covering codes in a similar way to how they are used in~\cite{AL00} to obtain a further improved SM upper bound for $\func{GAF}$. To that end, let us define an \emph{$(m, r)$-covering code of length $n$} as a set of $m$ codewords with the property that the Hamming balls of radius $r$ centered around the codewords cover $\zo^n$.

    \begin{lemma} \label{lem:covering-code-implies-efficient-BSM-protocol}
        If there is an $(m, r)$-covering code of length $n$, then for every function $g \co \zo^n \to \zo$, the function $g(x \vee y)$ admits a $2$-party, $O\left(m r \binom{n}{\leq r/2} \right)$-size BSM protocol.
    \end{lemma}

    \begin{proof}
        The proof is similar to the proof of Theorem 2 in~\cite{AL00}. Let $\set{c_1, \dotsc, c_m}$ be an $(m, r)$-covering code. For every $1 \leq i \leq m$ and $0 \leq j \leq r$, let us denote by $H_{i, j}$ the Hamming sphere of radius $j$ centered around $c_i$; that is, $H_{i, j} = \set{z \in \zo^n \co d(z, c_i) = j}$.

        For a fixed $1 \leq i \leq m$ and $S \seq [n]$, let us define the following function:
        \[
        \Phi_{i, S}(z) := \left( \bigwedge_{k \in S \co [c_i]_k = 0} z_k \right) \wedge \left( \bigwedge_{k \in S \co [c_i]_k = 1} \overline{z_k} \right);
        \]
        and the following DNF:
        \[
        g_{i, j}(z) = \bigvee_{|S| = j} \left( g(c_i + S) \wedge \Phi_{i, S}(z) \right).
        \]

        Clearly, the $g_{i, j}$ functions are unate, and one can easily verify that $g_{i, j}(z) = g(z)$ for every $z \in H_{i, j}$. This implies a BSM protocol for $g(x \vee y)$: Alice and Bob will simulate the protocol from Lemma~\ref{lem:narrow-monotone-implies-efficient-BSM-protocol} for every $1 \leq i \leq m$ and $0 \leq j \leq r$, and Carol will compute $x \vee y$ to find out in which Hamming sphere it lies and retrieve the correct output. Overall, using the bound given in Lemma~\ref{lem:narrow-monotone-implies-efficient-BSM-protocol}, we get a BSM protocol of size $O\left(m r \binom{n}{\leq r/2} \right)$.
    \end{proof}

    We will need the following lemma, which argues the existence of covering codes that we shall make use of.

    \begin{lemma}[\cite{AL00}] \label{lem:existence-of-covering-codes}
        For every $n$ and $r$, there exists an $(m, r)$-covering code with $m = O\left( n 2^n / \binom{n}{\leq r}\right)$.
    \end{lemma}

    We are now ready to prove the following proposition.
    \begin{theorem} \label{thm:BSM-upper-bound-for-all-bitwise-combined-functions}
        Let $g \co \zo^n \to \zo$ be any function. Then, $g(x \vee y)$ admits a $2$-party, $O(2^{0.729 n})$-size BSM protocol.
    \end{theorem}

    \begin{proof}
        Take the $(m, r)$-covering code from Lemma~\ref{lem:existence-of-covering-codes} and use the matching BSM protocol from Lemma~\ref{lem:covering-code-implies-efficient-BSM-protocol}. This gives us a BSM protocol of size $O\left(m r \binom{n}{\leq r/2} \right)$ with $m = O\left( n 2^n / \binom{n}{\leq r}\right)$. The remaining part now is purely algebraic and is already carried out in the proof of Theorem 4 in~\cite{AL00}, so we won't repeat it entirely and instead refer the reader to~\cite{AL00} for the details. We will just mention that setting $r = (1 - 1/\sqrt{2}) n$, yields the desired upper bound:
        \[
        O\left( m r \binom{n}{\leq r/2} \right) = O\left( r n 2^n  \cdot \left(\binom{n}{\leq r/2} \Big/ \binom{n}{\leq r}\right) \right) = \dotsb = O(2^{(0.728...+o(1)) n}). \qedhere
        \]
    \end{proof}

    \begin{remark}
        Theorem 3 in~\cite{AL00} shows that using Lemma~\ref{lem:covering-code-implies-efficient-BSM-protocol} we cannot obtain a better upper bound than $2^{0.728 n}$.
    \end{remark}

\subsection{\texorpdfstring{Degree bounded polynomials over $\Z_6$}{Degree bounded polynomials over Z6}} \label{ssec:degree-bounded-polynomials-over-Z_6}

	In the previous section, we obtained tight bounds on the degree of Carol computing the Boolean Equality function when Carol is a polynomial over $\Z_2$. In this section, we ask if there is any savings (in terms of degree) in computing the equality function if we allow Carol to be a polynomial over $\Z_m$ where $m$ is composite.
	
	We show that this is indeed the case for $m=6$. More precisely, we show that there is a BSM protocol with preprocessing by Alice and Bob with Carol being a degree $2$ polynomial over $\Z_6$ and a $\Z_6$-predicate $P$ such that $P$ evaluated on the output of Carol computes the equality function.

	\begin{proposition}\label{prop:BSM-via-matching-vector}
		There exists a Boolean decision predicate $P \colon \Z_6 \to \set{0,1}$ and a BSM protocol computing the Boolean equality function $\EQ$ on $2n$ bits where Alice and Bob outputs vectors in $\Z_6^k$ and Carol computes a degree $2$ polynomial over $\Z_6$ such that $\phi$ evaluated on Carol's output equals the equality function where $k = 2^{O(\sqrt{ n \log n})}$.
	\end{proposition}
 
	   This gives a constant degree Carol with sublinear preprocessing output length over $\Z_6$ which should be contrasted with the BSM protocol in Proposition~\ref{prop:universal-degree-reduction-upper-bound} of polynomial  preprocessing output length and near linear degree over $\Z_2$.  Our construction of the BSM protocol goes via the existence of matching vector family over $\Z_m$ for composite integer $m$. For two vectors $u,v \in \Z_m^k$, let $\langle u, v \rangle$ be given by $\sum_i u_iv_i \mod 6$. We start by defining a matching vector family. 
	
	\begin{definition}[Matching Vector Family] Let $S \subseteq \Z_m\setminus \{0\}$ and let $\mathcal{F} = (\mathcal{U}, \mathcal{V})$ where $\mathcal{U} = (u_1, \ldots, u_n)$ and $\mathcal{V} = (v_1,\ldots, v_n)$ such that for all $i \in \{1,2\ldots,n\}$  $u_i, v_i \in \Z_m^k$. The family $\mathcal{F}$ is called as an $S$-matching vector family of size $n$ and dimension $k$, if for all $i,j \in \{1,2\ldots,n\}$,
	\[ \langle u_i , v_j \rangle  \begin{cases}
	 = 0 & \text{ if } i = j \\
	 \in S & \text{ if } i \ne j
\end{cases}\]			
		
	\end{definition}
	
	Grolmusz~\cite{Gro00} constructed explicit matching vectors family over $\Z_m$ for any composite $m$. 
	
	\begin{proposition}[\cite{Gro00}, Theorem 1.2]
		Let $m$ be a positive integer expressible as a product of distinct primes, $m = p_1^{\alpha_1}p_2^{\alpha_2}\ldots p_r^{\alpha_r}$ with $r \ge 2$. Then, for an integer $c$ (depending on $m$) there is an explicitly constructible $S$-matching vector family $\mathcal{F}$ in $\Z_m^k$ of size $n \ge \exp \left (\frac{c(\log k)^r}{(\log \log k)^{r-1}} \right)$ where $S = \{a \in \Z_m \mid \forall i \in [r], a \mod p_i \in \{0,1\}\} \setminus \{0\}$.
	\end{proposition}

	For the special case when $r=2$ with $p_1 = 2$ and $p_2 = 3$, the following corollary holds.
	
	\begin{corollary}\label{corr:matching-vector-family-mod-6}
		There exists an explicitly constructible $S$-matching vector family $\mathcal{F}$ in $\Z_6^k$ of size $n \ge \exp\left (c\cdot(\log k)^2/\log \log k  \right )$ with $S = \{1,3,4\}$.
	\end{corollary}
	
	We are now ready to prove the main result of this section.
	
	\begin{proof}[Proof of Proposition~\ref{prop:BSM-via-matching-vector}]
		Before the start of the protocol, Alice and Bob constructs an $S$-matching vector family of size at least $2^n$ over $\Z_6$  which exists by  Corollary~\ref{corr:matching-vector-family-mod-6} with $S = \{1,3,4\}$. Towards this, they choose a $k$ such that 
		\begin{equation} \label{eq:bound-on-k}
		n \ln 2 \le \frac{c(\log k)^2}{\log \log k} \le 2 n \ln 2
		\end{equation}

		This gives an explicit matching vector family in $\Z_6^k$ given by $(\mathcal{U}, \mathcal{V})$ with 
		\[ |\mathcal{U}| = |\mathcal{V}| \ge \exp\left (c\cdot (\log k)^2/\log \log k  \right ) \ge 2^n \]
		The last inequality follows by the lower bound for $k$ in Eq.~\ref{eq:bound-on-k}. Let $int(x)$ be the unique integer corresponding to the binary string $x$. We now describe what each of Alice, Bob and Carol computes. 
		
		Alice computes $A$ which maps each input string $x \in \{0,1\}^n$ to $u_{int(x)} \in \mathcal{U}$. Similarly, Bob computes $B$ which maps each input string $y \in \{0,1\}^n$ to $v_{int(y)} \in \mathcal{V}$. The length of their output is $k$ which is bounded by $2^{O(\sqrt{n \log n})}$ by upper bound for $k$ in Eq.~\ref{eq:bound-on-k}.
		
		We now describe Carol. Let $p$ be a polynomial over $\Z_6$ on $2k$ variables and $q$ be a homomorphism from $\Z_6$ to $\F_2^2$ defined respectively as,
		\[ p(a,b) = \langle a, b\rangle = \sum_{i=1}^k a_ib_i \mod 6  \quad \quad \quad q(c)  = (c \text{ mod } 2, c \text{ mod } 3) \] 
		
		Carol on input $u \in \Z_6^k$ from Alice and $v \in \Z_6^k$ from Bob outputs $p(u,v)$ which is a polynomial of degree $2$ over $\Z_6$. 
		
		Let $P\colon \Z_6 \to \F_4$ be defined as $P(c) = q(c)^3+1$. Since $q(c)^3 \in \set{0,1}$, $P$ is indeed a Boolean predicate. In this definition, we are implicitly using the isomorphism between $\F_4$ and $\F_2^2$.
		We now argue the correctness. When $x = y$ with $t =int(x) = int(y)$, Alice and Bob outputs $u_t$ and $v_t$ respectively. By property of the matching vector family, $\langle u_t,v_t \rangle = p(u_t,v_t) = 0$. Hence, the output of Carol given by $C(A(x), B(y)) = p(u_t,v_t) = 0$. Now evaluating $P$ on Carol's output will give a $1$. When $x \ne y$ with $t_x = int(x)$ and $t_y = int(y)$, Alice and Bob outputs $u_{t_x},v_{t_y}$ respectively. This means that  $p(u_{t_x},v_{t_y}) = \langle u_{t_x},v_{t_y} \rangle  \in S$. Let $c \in S$ be the value obtained. Since $0 \not \in S$, $q(c)$ is guaranteed to be a non-zero element in $\F_4$. Hence, $P(c) = q(c)^3+1$ will evaluate to $0$.
	\end{proof}

\end{document}